\documentclass[11pt]{article}
\usepackage{amsmath}
\usepackage{amsthm}
\usepackage{amssymb}
\usepackage{algorithm}
\usepackage{subfig}
\usepackage{color}
\usepackage[english]{babel}
\usepackage{graphicx}
\usepackage{wrapfig,epsfig}
\usepackage{epstopdf}
\usepackage{url}
\usepackage{graphicx}
\usepackage{color}
\usepackage{epstopdf}
\usepackage{algpseudocode}
\usepackage{scrextend}
\usepackage[T1]{fontenc}
\usepackage{bbm}
\usepackage{comment}
\usepackage{authblk}
\usepackage{multicol}
\usepackage{enumitem}
\usepackage[usenames,dvipsnames]{xcolor}
\usepackage{tikz}
\usetikzlibrary{
  graphs,
  graphs.standard
}

\usepackage{stackengine} 

\usepackage{tikz}
\usepackage{hyperref}  
\hypersetup{colorlinks=true,citecolor=blue,linkcolor=blue} 
\usetikzlibrary{arrows}
\usepackage[margin=1in]{geometry}
\graphicspath{{./figs/}}

\author{
  Mahdi Cheraghchi\thanks{EECS Department, University of Michigan--Ann Arbor, Email:
  \texttt{mahdich@umich.edu}}
\quad 
  Vasileios Nakos\thanks{Saarland University, Email: \texttt{vnakos@mpi-inf.mpg.de}. This work is a part of the project TIPEA that has received funding from the European
Research Council (ERC) under the European Unions Horizon 2020 research and innovation programme (grant agreement No.\ 850979).}
}  

\date{}

\title{Combinatorial Group Testing and Sparse Recovery Schemes with Near-Optimal Decoding Time}
\newtheorem{theorem}{Theorem}[section]
\newtheorem{lemma}[theorem]{Lemma}
\newtheorem{definition}[theorem]{Definition}

\newtheorem{remark}[theorem]{Remark}
\newtheorem{claim}[theorem]{Claim}

\newcommand{\wt}{\widetilde}

\renewcommand{\varepsilon}{\epsilon}
\renewcommand{\tilde}{\wt}

\DeclareMathOperator{\supp}{supp}

\DeclareMathOperator{\poly}{poly}

\makeatletter
\newcommand*{\RN}[1]{\expandafter\@slowromancap\romannumeral #1@}
\makeatother

\usepackage{lineno}

\usepackage{etoolbox}

\makeatletter

\newcommand{\define}[4][ignore]{%
  \ifstrequal{#1}{ignore}{}{
  \@namedef{thmtitle@#2}{#1}}%
  \@namedef{thm@#2}{#4}%
  \@namedef{thmtypen@#2}{lemma}%
  \newtheorem{thmtype@#2}[theorem]{#3}%
  \newtheorem*{thmtypealt@#2}{#3~\ref{#2}}%
}

\newcommand{\state}[1]{%
  \@namedef{curthm}{#1}
  \@ifundefined{thmtitle@#1}{
  \begin{thmtype@#1}
    }{
  \begin{thmtype@#1}[\@nameuse{thmtitle@#1}]
  }
    \label{#1}
    \@nameuse{thm@#1}
  \end{thmtype@#1}
  \@ifundefined{thmdone@#1}{
  \@namedef{thmdone@#1}{stated}%
  }{}
}

\newcommand{\restate}[1]{%
  \@namedef{curthm}{#1}
  \@ifundefined{thmtitle@#1}{
    \begin{thmtypealt@#1}
    }{
  \begin{thmtypealt@#1}[\@nameuse{thmtitle@#1}]
  }
    \@nameuse{thm@#1}
  \end{thmtypealt@#1}
  \@ifundefined{thmdone@#1}{
  \@namedef{thmdone@#1}{stated}%
  }{}
}

\newcommand{\thmlabel}[1]{
  \@ifundefined{thmdone@\@nameuse{curthm}}{\label{#1}
    }{\tag*{\eqref{#1}}}
}
\makeatother

\begin{document}

\begin{titlepage}
  \maketitle
  \begin{abstract}
In the long-studied problem of combinatorial group testing, one is asked to detect a set of $k$ defective items out of a population of size $n$, using $m \ll n$ disjunctive measurements. In the non-adaptive setting, the most widely used combinatorial objects are disjunct and list-disjunct matrices, which define incidence matrices of test schemes. Disjunct matrices allow the identification of the exact set of defectives, whereas list disjunct matrices identify a small superset of the defectives.
Apart from the combinatorial guarantees, it is often of key interest to equip measurement designs with efficient decoding algorithms. The most efficient decoders should run in sublinear time in $n$, and ideally near-linear in the number of measurements $m$.

In this work, we give several constructions with an optimal number of measurements and near-optimal decoding time for the most fundamental group testing tasks, as well as for central tasks in the compressed sensing and heavy hitters literature. For many of those tasks, the previous measurement-optimal constructions needed time either quadratic in the number of measurements or linear in the universe size. 

Among our results are the following: a construction of disjunct matrices matching the best-known construction in terms of the number of rows $m$, but achieving nearly linear decoding time in $m$;
a construction of list disjunct matrices
with the optimal $m=O(k \log(n/k))$ number of rows and nearly linear 
decoding time in $m$; 
error-tolerant variations of the above constructions; 
a non-adaptive group testing scheme for the ``for-each'' model
with $m=O(k \log n)$ measurements and $O(m)$ decoding time;
a streaming algorithm for the ``for-all'' version of the heavy hitters problem in the strict
turnstile model with near-optimal query time, as well as a ``list decoding'' variant obtaining
also near-optimal update time and $O(k\log(n/k))$ space usage; an $\ell_2/\ell_2$ weak identification
system for compressed sensing with nearly optimal sample complexity
and nearly linear decoding time in the sketch length.

Most of our results are obtained via a clean and novel approach that avoids list-recoverable codes or related complex techniques that were present in almost every state-of-the-art work on efficiently decodable constructions of such objects.

  \end{abstract}
  \thispagestyle{empty}
\end{titlepage}

\newpage
\thispagestyle{empty} \tableofcontents \newpage \addtocounter{page}{-1}
\section{Introduction}


The study of combinatorial group testing dates back to the Second World War,
suggested by Dorfman \cite{ref:Dor43} in the context of testing blood samples collected
from a large population of draftees. In an abstract formulation, a population of
$n$ individuals contains up to $k$, for a known parameter $k$, defectives and
tests are conducted to identify the exact set of defectives. Each test identifies a
subset of the individuals and returns positive if and only if the set contains at least one
defective individual. The basic combinatorial goal is to minimize the number
of tests required to identify the exact set of defectives in the worst case.
This article focuses on non-adaptive tests where the
tests are all pre-determined and can be conducted in parallel. In this case, the test design can be identified by a binary incidence matrix with $n$
columns and one row per test. 

Since its inception, group testing has found countless uses both in theory
and practice. Practical applications include a wide range of areas such as
molecular biology and DNA library screening (cf.\
\cite{ref:BKBB95,cheng2008new,ref:KM95,ref:KKM97,ref:Mac99,ref:ND00,ref:STR03,ref:WHL06,ref:WLHD08}
and the references therein), Human Genome Project (cf.\ \cite[Section~VI.46]{ref:handbook}),
multiple access communication
\cite{ref:Wol85}, data compression \cite{ref:HL00}, pattern matching
\cite{ref:CEPR07}, secure key distribution \cite{ref:CDH07}, 
network tomography \cite{ref:GCGT12}, quality control \cite{ref:SG59}, among others. 
The reader is invited to consult \cite{ref:groupTesting,ref:DH06}
for a more comprehensive discussion of the application areas.
Finally, the original idea of using group testing for pooling samples in
medical tests has recently gained renewed interest during the COVID-19
pandemic due to the prevalent shortage of test kits (cf.\ \cite{ref:COVID5,ref:COVID6,ref:COVID7,ref:COVID1,ref:COVID2,ref:COVID3,ref:COVID4,nhl20}).

In theoretical computer science, group testing falls under the broader umbrella
of \emph{sparse recovery}, where the general framework deals with the recovery
of sparse structures (such as high-dimensional vectors with few nonzero entries 
or their approximations)  via queries from a restricted class (such as linear
queries, as in \emph{compressed sensing} \cite{ref:CSbook}, disjunctive queries which
define group testing, or by sampling Fourier coefficients of the underlying vector~\cite{CandesTao,hikp12a,k17,NakosSW19}). The general area of sparse recovery provides a fundamental
toolkit for the study of streaming and sublinear time algorithms, and technology from that area lies at the heart of the latest improvements for Subset Sum~\cite{ABJTW19,BN20} and Linear programming~\cite{lp20}. As a combinatorial construct,
sparse recovery, and more specifically group testing, is related to the notion of selectors \cite{ref:CGOR00} and related 
pseudorandom objects.
\paragraph{Disjunct Matrices.} The combinatorial guarantee for a test design to allow for the identification
of the set of defectives is studied in the literature under several essentially equivalent
notions, such as superimposed codes, cover-free families (or codes), and 
disjunct matrices (Definition~\ref{def:disjunct} 
(see \cite[Chapter~19]{ref:GRS19}, \cite[Chapter~4]{ref:groupTesting} and
\cite{dyachkov1983survey} for a detailed discussion). Roughly speaking, a disjunct test matrix for $k$ defectives satisfies the following:
for every set $S$ of $k$ columns and a column $i \notin S$, there is a row at which
the columns in $S$ have zeros whereas the $i$th column has a $1$. 
A lower bound of $\Omega(k^2 \log_k n)$ on the number of rows has been proved several times in the literature
\cite{d1982bounds,dyachkov1989superimposed,ruszinko1994upper}. The best-known construction achieves $m = O(k^2 \mathrm{min}\{\log n, (\log_k n)^2\} )$ number of rows
(by a combination of the Kautz-Singleton construction \cite{ref:KS64} and Porat-Rothschild \cite{porat2008explicit}). 
The notion of disjunctness can be naturally extended to also allow for accurate
recovery in presences of false positives and
negatives in the test outcomes.

Two central problems in group testing are explicit construction of test designs
and efficient recovery of the defectives from test outcomes. While a simple probabilistic
argument can achieve an upper bound of $O(k^2 \log n)$ tests (cf.\ \cite[Chapter~4]{ref:groupTesting}), an explicit construction (in polynomial time in the matrix size) matching this upper bound \cite{porat2008explicit} can be significantly more challenging. From the recovery
perspective, any disjunct matrix allows recovery in nearly linear time in the size
of the matrix using the following \emph{naive decoder}: the decoder can simply output the subset of the columns of the test matrix
whose supports are contained in the support of the test outcomes. For large population
sizes, however, it is desirable to have a \emph{sublinear time} recovery 
algorithm that runs in polynomial time
(or even nearly linear time) in the number of tests, which can potentially be exponentially
faster than the naive decoder above. 

\paragraph{List-Disjunct Matrices.} Another important combinatorial object, introduced independently in \cite{ch2009noise, indyk2010efficiently}, is that of a list-disjunct matrix. List-disjunct matrices guarantee the recovery of a small superset of the defective items but feature the advantage that the number of rows can be much smaller than what a disjunct matrix would allow (essentially by a factor of $k$), among several additional notable advantages and applications. Using list-disjunct matrices, one can design \emph{two-stage} group testing schemes, by first narrowing the universe down to a small set, and then performing a test on each one independently. Thus, in scenarios where two-stage testing is possible, for example in DNA library screening or data forensics~\cite{goodrich2005indexing}, this results in a major savings. Moreover, list-disjunct matrices can be used for constructions of monotone encodings and multi-user tracing families~\cite{alon2009optimal}, vote storage systems~\cite{moran2007deterministic}, and for designing state of the art heavy hitter sketches (as we show in this work). Last but not least, they can be used as an intermediate tool towards the construction of (efficiently decodable) disjunct matrices; indeed this was the main motivation in~\cite{ch2009noise, indyk2010efficiently}.

At least for noiseless testing, there is a simple \emph{bit masking} trick that can augment
any disjunct or list-disjunct matrix with additional rows to enable sublinear recovery
(e.g., see \cite{lee2019saffron}). The augmentation
blows up the number of rows by a logarithmic factor in $n$, and thus a long line of work has been devoted to obtaining better trade-offs between rows and recovery time.

From now on, we shall refer to \emph{decoding time}, as the time needed for recovery of the defectives or a small list containing them. We also stress the difference between the ``for-all'' guarantee (\emph{uniform}) and the ``for-each'' guarantee (\emph{non-uniform}). A matrix satisfies the first guarantee if it enables recovery for all vectors simultaneously, while a randomized matrix (i.e., a distribution over matrices) satisfies the for-each guarantee if it enables recovery of a fixed vector with some target probability. Disjunct and list-disjunct matrices are defined with the for-all guarantee in mind.


\paragraph{Work on sublinear-time group testing and related problems.}
Sublinear-time decoding on group testing (including disjunct, list-disjunct matrices, the probabilistic and the non-uniform case) has been explored in~\cite{GuruswamiI04,ref:CM05,chen2008survey,indyk2010efficiently,icalp11,cai2017efficient,vem2017group,lee2019saffron,bondorf2019sublinear}. In the context of the similar tasks of heavy hitters and compressed sensing (see below), sublinear-time has been investigated in~\cite{gms05,glps12,ps12,hikp12a,hikp12b,gnprs13,lnnt16,k16,NakosOne17,Nakos17,k17,cksz17,glps17,LiN18,NakosS19}, to name a few.

There is also a decent amount of literature on variants of the group testing problem, such as sparse group testing~\cite{GGJZ19}, graph-constrained group testing~\cite{Cheraghchi13,SpangW19}, and threshold group testing~\cite{Cheraghchi13}. Our focus in this paper is the most standard setup of the problem, although our techniques could potentially apply to the aforementioned settings as well.
\paragraph{Heavy Hitters and Compressed Sensing.} A closely related problem is the task of finding heavy hitters in data streams. Given a long stream of updates $(i,\Delta)$ to a vector $x \in \mathbb{R}^n$ causing $x_i \leftarrow x_i + \Delta$, upon query detect the coordinates $i \in [n]$ which satisfy $|x_i| \geq (1/k) \|x\|_p$ (heavy hitters). The goal is to keep a small-space representation of $x$ which allows finding the heavy hitters quickly, as well as rapid updates. The most interesting and well-studied cases correspond to $p=1$ and $p=2$. The heavy hitter problem is one of the core problems in streaming algorithms and has also served implicitly or explicitly as a subroutine in many streaming and compressed sensing algorithms; cf.\ ~\cite{gms05,HNO08,KaneNPW11,hikp12a,hikp12b,iw13,ik14,glps17,JW18} to name a few. It has also been an active area of research with many important results being discovered in the 2000s \cite{ccf02,cm04,ch09}, as well as more recently \cite{bciw16,w16,bcinww17,lnnt16,nakos2017deterministic_heavy_hitters,aamand2019learned,BGLWZ18,BDW19}.

Another closely related area is \emph{compressed sensing}~\cite{CandesTao,d06,glps12,hikp12a}, which focuses on understanding the design of a set of linear measurements $\Phi \in \mathbb{R}^{m \times n}$, such that given $y = \Phi x$ it is possible to recover an approximation to the best $k$-sparse approximation of $x$ with respect to some $\ell_p$ norm. This problem is analogous to the heavy hitters problem, albeit with the difference that one desires to recover \emph{most} of the heavy hitters in an $\ell_p$ sense, rather than all of them. Since the literature on the topic is vast, we refer the reader to a survey of Indyk and Gilbert \cite{gi10}, the introduction in \cite{NakosS19}, and the text \cite{ref:CSbook}.
Henceforth group testing, heavy hitters, and compressed sensing may be referred to using the umbrella term \emph{sparse~recovery}.

\paragraph{Our Contributions.} We give several schemes for the sparse recovery problem, almost all of which feature near-optimal (nearly-\emph{linear}) decoding time, improving upon several results in the literature, and setting the record straight for some of the most well-studied variants of the problem. We thus show that previous trade-offs in measurement complexity and decoding time can be greatly improved. In particular, we contribute the following.

\begin{itemize}
\item Combinatorial Group Testing
\begin{enumerate}
\item A Monte-Carlo construction of list-disjunct matrices with the optimal $O(k \log (n/k))$ number of rows and $O(k\log^2(n/k))$ decoding time. The best previous sublinear-time scheme in terms of measurements suffered from quadratic decoding time in $k$ and did not achieve the optimal number of rows. We thus essentially settle the measurement and the decoding time complexity of list-disjunct matrices.
\item A Monte-Carlo construction of $k$-disjunct matrices with $m = O(k^2 \mathrm{min}\left\{ \log n, (\log_k n)^2 \right\})$ rows and $O(m + k \log^2 (n/k))$ decoding time. 
Moreover, our construction can use an off-the-shelf construction of disjunct matrices
(which may have an inefficient decoder) as a black box, so any improvement on the
construction of disjunct matrices will immediately improve our result as well, 
resulting in a construction of disjunct matrices with the same number of rows and near-optimal decoding time.

\item An explicit construction of $k$-disjunct matrices with $m=O(k^2 \log n)$ rows with decoding time nearly linear in $m$. 
\item State-of-the-art error-correcting disjunct and list-disjunct matrices, associated with decoding procedures which are faster by almost a factor $k$ from previous schemes with the same number of rows.
\item A (necessarily randomized) scheme with $O(k \log n)$ decoding time and measurements for the ``for-each'' version of the group testing problem, improving upon recent work which obtained the same number of measurements but with quadratic in $k$ decoding time. This result essentially settles the ``for-each'' complexity of the group testing problem.
\end{enumerate}
\item Heavy Hitters
\begin{enumerate}[resume]
\item A ``for-all'' streaming algorithm with $s = O(k\log(n/k))$ space usage for the heavy hitters problems in the strict turnstile model, allowing finding a list of size $O(k)$ that contains all $(1/k)$-heavy hitters. The query time is near-linear in $s$ and the update time is $\tilde{O}(\log^2k\cdot \log(n/k))$. In contrast, the previous algorithm with the same space required $\Omega(n \log n)$ query time and $\Omega(k \log (n/k))$ update time.
\item A ``for-all'' streaming algorithm for the standard version of the heavy hitters problem in the strict turnstile model, matching the space usage $s$ of previous constructions and allowing queries in time near-linear in $s$. Previous constructions suffered from $\Omega( n k )$ query time.
\end{enumerate}
\item Compressed Sensing
\begin{enumerate}[resume]
\item A significantly stronger $\ell_2/\ell_2$ weak identification system than what was available before in the compressed sensing literature.
\end{enumerate}
\end{itemize}

The most efficient previous sublinear-time schemes employed list-recoverable codes technology and the list-decoding view of pseudorandom objects such as expanders and extractors, or related ideas. On the other hand, most of our results stem from a 
unifying result (Theorem~\ref{thm:main_tool}) which morally is the following: ``There exists a row-optimal $(k,5k\log(n/k))$-list-disjunct matrix associated with a very efficient decoding procedure''. Interestingly, in contrast to list-recoverable codes type of arguments which come with relatively large constants and many parameters to fine-tune, the aforementioned result and its implications require a minimal understanding of coding theory, being of potential practical impacts.

We also bring the readers' attention to the concurrent work of Price and Scarlett~\cite{PriceS20}, which arrives at the construction of our efficiently decodable list-disjunct matrix (without analyzing the tolerance of the matrix on false positives) with a nearly identical algorithm, and use it to obtain $O(k \log n)$ time for the ``for-each'' version of group testing, matching our Bullet 5. Their analysis of the decoding algorithm is quite different, relying on bounds for sub-exponential random variables to control the branching process created by the execution of the algorithm. On the other hand, our argument is elementary and is based only on first principles. An advantage of their $O(k \log n)$-time algorithm is that they are able to guarantee correctness using limited independence for the hash functions~\cite[Section~3]{PriceS20}, obtaining thus a low-storage and still efficient variant of the $O(k \log n)$-time algorithm (our Bullet 5).

\section{Preliminaries}
\label{sec:prelim}
\subsection{Notation}
\newcommand{\bpref}{\mathsf{bPref}}

When referring to group testing, all matrices and vectors have entries in $\{0,1\}$, with $0$ corresponding to ``false'' and $1$ to ``true''. Without loss of generality, we can assume that $k,n$ are powers of two by rounding to the closest power of two from above, unless noted otherwise. We will associate $[n] := \{0, \ldots, n-1\}$ with $\{0,1\}^{\log n}$ via the obvious bijection. Moreover, all matrices and vectors are zero-indexed, that is for vector $x \in \{0,1\}^n$ the entries are $x_0,x_1,\ldots,x_{n-1}$.
More generally, for a set $\mathcal{I} \subseteq [n]$,
we denote by $x_{\mathcal{I}}$ the vector obtained by discarding the entries of $x$ outside $\mathcal{I}$. We also denote the \emph{support} of a vector by $\mathrm{supp}(x) = \{ i \in [n]\colon x_i =1 \}$. 
For binary strings $r,s$, we write $r\| s$ to be the concatenation of $r$ and $s$ by writing $r$ followed by $s$.
For a test matrix $ M \in \{0,1\}^{m \times n}$ (where the number of columns $n$ is called the \emph{population} or the \emph{universe} size), and a vector $x \in \{0,1\}^n$, the vector $ y = M \odot x $ corresponds to tests 
\begin{align*}
y_q = \bigvee_{j \in [n]\colon M_{q,j} =1 } x_j, \forall q \in [m],
\end{align*}
where $M_{q,j}$ denotes the entry of $M$ at row $q$ and column $j$.

For $ i \in [n]$ we denote by $M^i$ the $i$-th column of $M$ and $M_i$ to be the $i$-th row of $M$. For $S \subseteq [n]$, we define $M^S = \bigcup_{i \in S} M^i$. Clearly $M^i = M^{\{i\}}$. 

When referring to heavy hitters or compressed sensing we will work with the standard notion of addition and multiplication on numbers of $\Theta(\log n)$ bits. We will say that $i$ is a $(1/k)$-heavy hitter for the vector $x \in \mathbb{R}^n$ if $|x_i| \geq (1/k)\|x\|_1$. We define $\|x\|_p = \left( \sum_{i=0}^{n-1} |x_i|^p \right)^{1/p}$. 
We denote by $x_{-k}$ the \emph{tail} vector that occurs after zeroing out the $k$ largest in magnitude coordinates in $x$. 

For non-negative integers $\alpha, \ell$ such that $\alpha \leq 2^{\ell}-1$, we denote by $\bpref_{\ell}(\alpha)$ to be the integer that is obtained from the first $\ell$ bits in the binary representation of $\alpha$. For example, $\bpref_{2}( (1100)_2 ) = (11)_2 = 3$, and $\bpref_3( (11011)_2 ) = (110)_2 = 6$, where we have used the notation $(\cdot)_2$ for the  binary representation of an integer).

\subsection{Catalan Numbers}
We shall use the following fact on Catalan numbers.

\begin{lemma}[generalized Catalan numbers, \cite{sloane2007line}]\label{lem:estimation_catalan}
For natural integers $d,n\geq 2$, the number of rooted $d$-ary trees with exactly $n$ nodes is 
\[\mathrm{Cat}_n^d = \frac{1}{n+1}  {d n\choose n} \leq (ed)^n.\]
\end{lemma}

\subsection{Disjunct and List-Disjunct Matrices}

In this section, we review the standard notion of disjunctness and
its variations that are instrumental
for the design of group testing schemes (cf.\ \cite[Chapter~4]{du2000combinatorial}).

\begin{definition}[Disjunct Matrices] \label{def:disjunct} 
A matrix $M \in \{0,1\}^{m \times n}$ is called $k$-disjunct if for every set $S \subseteq [n]$ of size $k$, and every $j \in [n] \setminus S$ there exists a row $q \in [m]$ such that $M_{q,j} =1 $ and $M_{q,j'}=0, \forall j' \in S$.
\end{definition}

A $k$-disjunct matrix essentially characterizes the combinatorial guarantee
needed for noiseless group testing. The relaxes notion of \emph{list-disjunct}
matrices guarantees identification of a bounded-sized superset of
the defective (thereby allowing a smaller number of rows by 
only requiring the recovery of 
a small list that is guaranteed to contain all defectives). The following
definition is from \cite{indyk2010efficiently} (while an essentially equivalent
notion was formulated in \cite{ch2009noise}).

\begin{definition}[List-Disjunct Matrices] 
A matrix $M \in \{0,1\}^{m \times n}$ is called a $(k,\ell)$-list-disjunct matrix if for every two disjoint sets $S,T \subseteq [n]$ with $|S|=k, |T|=\ell+1$ there exists an element $j \in T$ and a row $q \in [m]$ such that 

	\[	M_{q,j} =1~\mathrm{and}~\forall j'\in S, M_{q,j'}=0.		\]
\end{definition}

The parameter $\ell$ captures the \emph{list size} (so that it is always possible
to output a list of size at most $k+\ell$ that contains the defective). 
For $\ell=0$ the notion of list-disjunct matrices
coincides with the classical notion of disjunct matrices (Definition~\ref{def:disjunct}).
Given the measurement outcomes, one can naturally consider a list of
possible defectives by selecting all columns of $M$ the are covered
by the vector of the measurement outcomes.  More precisely,
we can define the following.

 \begin{definition}[Associated List for List-Disjunct Matrices]
Given $y = M\odot x$ with $M$ being $(k,\ell)$-list-disjunct matrix and $|\mathrm{supp}(x)| \leq k$, we will refer to $L$ as the \emph{associated list} of $x$ with respect to $M$ as the list of elements $i \in [n]$ satisfying the following:
	\[	\forall q \in [m] \text{ such that } y_q = 1\colon M_{q,i} = 1.\] 
Put simply, $L$ corresponds to the elements $i \in [n]$ which appear to be ``defective'' under measurements defined by $M$. Note that $|L| \leq k+\ell$.
\end{definition}
The above notions can be strengthened to tolerate errors as follows:


\begin{definition}\cite[Definition~1]{icalp11}
A matrix $M \in \{0,1\}^{m \times n}$ is called $(k,\ell,e_0,e_1)$-list disjunct if for every 
disjoint sets $S, T \subseteq [n]$ of size $k$ and $\ell$, respectively, the following
holds. Let $M^S$ and $M^T$ respectively denote the unions of supports of the 
columns of $M$ picked by $S$ and $T$. Then, for 
every set $X \subseteq M^T \setminus M^S$ of size $|X| \leq e_0$,
there is a column $M^j$ picked by $T$ such that
$|\supp(M^j) \setminus (X \cup M^S)| > e_1$.
\end{definition}

In the above definition, $e_0$ (resp., $e_1$) captures the number of false positives
(resp., negatives) that the matrix $M$ can combinatorially 
tolerate in the measurement outcomes (and in the sequel, this is what we 
would mean by a matrix tolerating a certain number of false positives or negatives).
We could have alternatively used the notion of error-correcting disjunct matrices in (\cite[Definition~1]{ch2009noise}), but since our results behave differently in the case of false positives and false negatives, the definition in~\cite{ch2009noise} is not the most suitable for our needs. We refer the reader to Theorems \ref{thm:error_list}~and~\ref{thm:error_dis} for the results on error-correcting $k$-disjunct matrices.	

It is shown in \cite[Proposition~2]{icalp11} that any $(k,\ell,e_0,e_1)$-list disjunct matrix 
guarantees recovery of a set of size less than $k+\ell$ containing all defective
items in presence of up to $e_0$ false positives and $e_1$ false negatives in the
test outcomes. Lower bounds in~\cite{ch2009noise,icalp11} show that $(k,\Theta(k),e_0,e_1)$-error-correcting list-disjunct matrices require $\Omega( k\log(n/k) + e_0 + ke_1)$ rows. Similarly, any $k$-disjunct matrix that can tolerate $e_0$ false positives and $e_1$ false negatives requires $\Omega( k^2\log_k n + e_0 + ke_1)$ rows.

A natural decoder for disjunct and list-disjunct matrices is the following, often referred to as the \emph{naive~decoder}.
\begin{definition}[Naive Decoder]\label{def:naive}
Given $y = M \odot x$, for every $i \in [n]$ declare $i$ defective if and only if $y_r=1$ for every $r$ such that $M_{r,i}=1$. That is, output the set of columns that
are covered by the measurement outcomes.
\end{definition}
The following are two well-known corollaries on the performance of the naive decoder in disjunct and list-disjunct matrices. 
\vspace{-2mm}
\begin{lemma}[Naive Decoder and Point-Queries for Disjunct Matrices] \label{lem:point_query_disjunct}
Given $y= M\odot x$, where $M$ is a $k$-disjunct matrix and $|\mathrm{supp}(x)| \leq k$ the following holds. Given $i \in [n]$, we can decide in time $O(|\mathrm{supp}(M^i)|)$ whether $i$ is defective or not. Moreover, the naive decoder returns $\mathrm{supp}(x)$ in time $O(n \cdot \mathrm{max}_{i \in [n]}|\mathrm{supp}(M^i)|)$.
\end{lemma}

\begin{lemma}[Naive Decoder and Point-Queries for List-Disjunct Matrices] \label{lem:point_query_list_disjunct}
Given $y = M\odot x$, where $M$ is a $(k,\ell)$ list-disjunct matrix, and $|\mathrm{supp}(x)| \leq k$ the following holds. Given $i \in [n]$ we can decide whether $i$ belongs to the associated list $L$ in time  $O(|\mathrm{supp}(M^i)|)$ whether $i$ is defective or not. Moreover, the naive decoder returns $L$ in time $O(n \cdot \mathrm{max}_{i \in [n]}|\mathrm{supp}(M^{i})|)$.

\end{lemma}

We shall use the following well-known constructions of $k$-disjunct matrices, which follows from standard constructions of incoherent matrices (based on either Reed-Solomon codes by Kautz and Singleton \cite{ref:KS64} or codes on the Gilbert-Varshamov bound \cite{porat2008explicit} by Porat and Rothschild)\footnote{%
The construction of \cite{ref:KS64} is strongly explicit, in the sense that each entry
of the matrix can be computed in $\poly(k, \log n)$ time, whereas \cite{porat2008explicit}
is explicit in the sense of being computable in $\poly(n)$ time.}

\begin{theorem}[Disjunct Matrices]\label{thm:disjunct_tool}
There exists a $k$-disjunct matrix with 
\[m = O(k^2 \mathrm{min}\left\{\log n, (\log_k n)^2 \right\})\]
rows. In particular, there exist explicit $k$-disjunct matrices with i) $O(k^2 \log n)$ rows and $O(k \log n)$ non-zeros per column (via \cite{porat2008explicit}), and strongly explicit
$k$-disjunct matrices with ii) $O(k (\log_k n)^2 )$ rows and $O( k \log_k n)$ non-zeros per column (via \cite{ref:KS64}).
\end{theorem}


We will use the following existential bound on list-disjunct matrices, that can be
derived from \cite{ch2009noise} via a probabilistic argument: 
\begin{theorem}[\cite{ch2009noise}]\label{thm:list_disjunct_tool}
There exists a $(k,k)$-list-disjunct matrix with $O(k \log (n/k))$ rows and $O(\log(n/k))$ non-zeros per column.
\end{theorem}

\section{Results}
\newcommand{\CFP}{{C_{\mathsf{FP}}}}
\newcommand{\CL}{{C_{\mathsf{L}}}}

In this section, we present our results on disjunct matrices, list-disjunct matrices, group testing, and heavy hitters, based on the definitions given in the preliminaries. In what follows, $C, \CL,\CFP> 1 $ are absolute constants. All our results assume, without loss of generality that $k \leq \gamma n$ for some absolute constant $\gamma$, otherwise storing the identity matrix is asymptotically the best solution. Our starting point and one of our strongest tools is the following theorem.

\begin{theorem}\label{thm:main_tool}
There exists a Monte-Carlo construction of a $(k, \CL k\log(n/k)  )$-list-disjunct matrix $M \in \{0,1\}^{m\times n}$ with $m = C \cdot k\log(n/k)$, allowing decoding in time $O(k \log (n/k))$. $M$ is the vertical concatenation of $\log(n/k)$ matrices $M^{(\log k)},\ldots, M^{(\log n)}$ such that (i) each such submatrix has $C k$ rows and exactly $1$ non-zero per column, (ii) every such submatrix can tolerate up to $\CFP \cdot k$ false positives.

Furthermore, $M$ can be stored in $O(k \log (n/k))$ space, and for every $\ell$ and every choice of $B = O(k\log(n/k))$ columns $i_1,i_2,\ldots,i_B \subseteq [n]$, we can find the rows $q_{i_1}, \ldots, q_{i_B} \subseteq [Ck]$ where the aforementioned columns have the non-zero element in $M^{(\ell)}$ in time 
\[	O\left( k \log^2 \left(\frac{n}{k}\right) \cdot \log^2( k \log\left(\frac{n}{k}\right)) \cdot \log \log\left( k \log \left(\frac{n}{k}\right)\right) \right).\]	
\end{theorem}

The last sentence of the above theorem, namely the claim about storing the matrix $M$ in small space and the fast batch location, is particularly important for our application to the heavy hitters problem. For the group testing applications, this property will be irrelevant. The matrix $M$ will be also called the \emph{identification} matrix.

\subsection{Disjunct and List-Disjunct Matrices}

\begin{theorem}[List-Disjunct Matrices]\label{thm:list_disjunct}
There exists a Monte-Carlo construction of a $(k,k)$-list-disjunct matrix $M \in \{0,1\}^{m \times n}$ with $m = O(k \log (n/k))$, that allows decoding in $O(k \log^2(n/k))$ time\footnote{For ease of exposition, we chose to give a construction of list-disjunct matrices with $\ell = k$.}.
\end{theorem}

In comparison, the best previous construction of efficiently decodable list-disjunct matrices requires either $k^2 \poly(\log n)$ decoding time and $O(k \log n \cdot \log \log_k n)$ rows \cite{icalp11}, or $O(k\log^2(n/k))$ rows and decoding time (\cite{icalp11} gives another construction using Parvaresh-Vardy codes with much less clean time and measurement bounds and polynomial in $k$ decoding time). 

\begin{theorem}[Disjunct Matrices]\label{thm:disjunct}
There exists a Monte-Carlo construction of a $k$-disjunct matrix $M \in \{0,1\}^{m \times n}$ with $m = O(k^2 \mathrm{min}\{ \log n, (\log_k n)^2 \} )$, that allows decoding in $O(m + k \log^2 (n/k))$ time.
\end{theorem}

The best previous constructions of efficiently decodable disjunct matrices are: i) The results from \cite{indyk2010efficiently}, which achieves $m=O(k^2 \log n)$ rows and $\Omega(k^4 \log n)$ decoding time, ii) The result from \cite{icalp11} which achieves $O(k^2\log n + k \log n \cdot \log \log_k n)$ rows and $m\log^2 n$ decoding time. We strictly improve upon the measurement and the decoding time complexity of previous work, obtaining the cleanest bounds.

\begin{theorem}[Explicit Disjunct Matrices]\label{thm:deterministic_disjunct}
We can construct in polynomial time in $n$ a $k$-disjunct matrix with $m= O(k^2 \log n)$ rows that allows decoding in time $m\cdot \poly(\log n)$, unless 
\[k \in \left[ \frac{C\log n}{\log \log n}, \left(\frac{C\log n}{ \log \log n}\right)^{1+o(1)}\right],\] where the $o(1)$ term is $\Theta\left(\frac{(\log \log \log n)^2}{\log \log n}\right)$.
\end{theorem}

Of course, the small intermediate range of $k$ where the above result does not apply
can be eliminated by slightly rounding $k$ up to $k^{1+o(1)}$, resulting in $m=k^{2+o(1)} \log n$ rows and $m \cdot \poly(\log n)$ decoding time for all $k$.

\subsection{Heavy Hitters in the Strict Turnstile Model}

\begin{theorem} (``For-all'' Heavy Hitters) \label{thm:heavy_hitters}
There exists a streaming algorithm with space usage $O(k \log (n/k))$, which keeps a (non-linear) representation of a vector $x \in \mathbb{R}^n$, and upon query, if $x \in \mathbb{R}^n_+$ then \emph{always} returns a list $L$ of size $O(k)$ which contains every $(1/k)$-heavy hitter. The query time is $O(k \poly( \log n))$ and the update time is $\tilde{O}( \log (n/k) \cdot \log^2 k )$.

\end{theorem}

In contrast to the result appearing in \cite{lnw17} which achieved $\Omega(n \log n)$ query time and $O(k \log (n/k))$ update time, our algorithm achieves nearly optimal query and update time. The non-linearity of the sketch does not play a role in the number of measurements, but only to achieve the desired update time. It is shown in \cite[Theorems 4,5]{lnw17} that if we drop the assumption of the strict turnstile model or additionally demand accurate estimates (up to $(1/k) \|x\|_1$) of the coordinates in $L$, then there exists no such linear sketch unless it has $\Omega(k^2)$ rows.

The next result is a streaming algorithm for the more common version of the heavy hitters problem, where one wants to find every $(1/k)$-heavy hitter and no $i$ with $x_i \leq (1/(2k)) \|x\|_1$. This greatly improves upon the scheme appearing in \cite{NNW12} which has the same space usage but requires $\Omega( n k )$ query time\footnote{The results in \cite{NNW12} satisfy a stronger guarantee, referred to as the ``tail'' guarantee in the sparse recovery literature. It is not hard to see that our arguments can facilitate that guarantee as well, but for ease of exposition we chose to present only the more standard guarantee of the heavy hitters problem.}.

\begin{theorem} (``For-all'' Heavy Hitters with Estimates) \label{thm:heavy_hitters_estimates}
There exists a streaming algorithm using space usage
\[	O\left(k^2 \cdot \mathrm{min}\left\{ \log n , \left( \frac{\log n}{\log k + \log \log n}\right)^2 \right\} \right),	\]
 which keeps a (non-linear) representation of a vector $x \in \mathbb{R}^n$, and upon query, if $x \in \mathbb{R}^n_+$ then \emph{always} returns a list $L$ containing every $(1/k)$-heavy hitter, and no $i\in [n]$ with $x_i \leq (1/(2k)) \|x\|_1$. Moreover, for every $i \in L$ it returns an estimate $x_i'$ with $ x_i \leq x_i' \leq x_i + (c/k)\|x\|_1$, where $c$ is an arbitrarily small absolute constant $c<1$.
The query time is $k^2 \poly( \log n)$ and the update time is $O(k \cdot \mathrm{min}\left\{ \log n ,\left( \frac{\log n}{\log k + \log \log n}\right)^2 \right\} ) + \tilde{O}(  \log ^3n  )$.

\end{theorem}

\begin{remark}
One could also ask whether the update time of $k$ on the above theorem is necessary, or more interestingly, one can decode $k$-disjunct matrices and perform queries for heavy hitters faster than quadratic time in $k$. After all, as will be revealed in Section~\ref{sec:obtain}, we need to point query only $O(k )$ coordinates, so it is not immediately evident that the quadratic time-bound in $k$ is necessary (we might need $\Omega(k^2)$ measurements, but an algorithm might not need to read all of them). However, it seems that performing point-queries is indeed a bottleneck, since i) an easy argument (which we leave to the reader) shows that any $k$-disjunct matrix must have at least $n-m$ columns of sparsity at least $k$, and ii) any $(1/k)$-incoherent matrix (from which known heavy hitters sketches follow) must have column sparsity $\Omega(k)$ as long as $m \leq n / \log k$ \cite[Theorem~10]{nelson2013sparsity}. This constitutes strong evidence that it is impossible to beat quadratic decoding/ query time and linear (in $k$) update time, unless using a near-linear in $n$ number of measurements.
\end{remark}

It is also worth noting that any subsequent improvement of sketches that enable $\ell_1$ point-queries immediately translates, via our framework, to a streaming algorithm with sublinear query time. Thus, we may consider the problem of sublinear-time query time essentially closed, up to logarithmic factors.

\subsection{Error-Correcting Disjunct Matrices}

We give the following two constructions of efficiently decodable matrices. The first is particularly efficient for false negatives, while the second for false positives. Both results are significantly faster than what was attainable by previous techniques using the same number of rows. We find it quite interesting that while we are able to construct fast error-correcting disjunct matrices with respect to either false positives or false negatives, we cannot construct fast error-correcting disjunct matrices that can facilitate both simultaneously, see an explanation in Section~\ref{sec:error}. It would be interesting to have a combination of these methods which could give the best of both worlds.

\begin{theorem}(False Positives)\label{thm:error_list}
There exist Monte-Carlo constructions of
\begin{enumerate} 
\item A $(k,k,e_0,0)$-list-disjunct matrix with $m = O(k \log (n/k) + \log_k n \cdot e_0 )$ rows which allows decoding in time $O(k^{\alpha} \cdot m )$, for any constant $\alpha>0$.
\item A $k$-disjunct matrix with $m = O(k^2 \log n + \log_k n\cdot  e_0)$, which can tolerate up to $e_0$ false positives and allows decoding in time $O(m \cdot \log n))$.
\end{enumerate}
\end{theorem}

The first result of the preceding theorem improves by almost a $k$ factor what is attainable by the techniques in \cite{icalp11,bondorf2019sublinear} (see also the comment following); the techniques in \cite{cai2013grotesque,cai2017efficient,lee2019saffron} result in schemes with a strictly larger number of rows\footnote{In fact, those results can be immediately obtained by the standard construction augmented with the code in \cite{s96}, see also discussion in Section~\ref{sec:tech}.}. For both results in Theorem~\ref{thm:error_list}, the argument in \cite{icalp11} obtains near-linear decoding albeit with a slight loss of $\log \log_k n$; \cite{cai2013grotesque,cai2017efficient,lee2019saffron} can be modified to obtain near-linear decoding but with a $\log n$ factor overhead in the measurement complexity.

\begin{theorem}(False Negatives) \label{thm:error_dis}
There exists a Monte-Carlo construction of a $k$-error-correcting disjunct matrix achieving $m = O(k^2 \log n + k\cdot  e_1)$, which can tolerate up to $ke_1$ false negatives and allows decoding in time $O(m \cdot \poly(\log n))$.
\end{theorem}

This theorem improves upon what was known and achievable using previous techniques, both in terms of measurements and decoding time, and achieves the optimal dependence in terms of $e_1$, the number of false negatives. 

\paragraph{Comment on the conversion in \cite{icalp11}.} The authors in \cite{icalp11}, using a black-box conversion (Section~3 and subsection~5.2), claim a construction of $(k,n,e_0,e_1)$-list-disjunct matrices with $m=O(k \log n \cdot \log \log_k n + \log \log_k n \cdot e_0 + k \log \log_k n \cdot e_1)$ rows and $m^2 \poly(\log m)$ decoding time; this can be derived by Corollary~12 in that paper by plugging in an optimal construction of error-correcting list-disjunct matrix (in the same spirit as Corollaries 13~and~14, but without the demand of explicitness). However, we note that Theorem~11 in that paper is incorrect (in that the disjunctness property with the given adversarial error bounds $e_0, e_1$ indeed hold combinatorially, but the given decoding algorithm cannot attain the reported error tolerance), and thus Corollaries 12,~13, and 14, as well as Corollary~22 and Theorem~24 are not fully correct with respect to all parameters. The reason is that in Theorem~11 the authors write down the recursive relation in terms of the universe size, but it should also incorporate the number of false positives and false negatives; i.e., $t(i,e_0,e_1)$ instead of $t(i)$. This translates to the existence of a $(k,n,e_0,e_1)$-list-disjunct matrix with $m=O( k \log n \cdot \log \log_k n + \log_k n \cdot e_0 + \log_k n \cdot e_1)$ rows instead; the $\log \log_k n$ factors in the false positives and false negatives in Corollaries 12,~13, and 14 (which are written down as $\log \log n$ factors) should be $\log_k n$ instead, and the $\log \log n$ factors in Corollary~22 and Theorem~24 should be $\log n$. To be more descriptive, the conversion in Theorem~12, Corollary~12 from \cite{icalp11} builds a binary recursion tree with $\log_k n$ nodes, and uses $O(k \log(n^{1/2^i}/k) + e_0 + ke_1)$ rows for every node in the $i$-th level of the tree; thus, the factor $k \log n$ is indeed multiplied by the number of levels; i.e., $\log \log_k n$, but the factor that multiplies $e_0+ke_1$ is the \emph{total~number~of~nodes}, which is $\log_k n$. Another way to see why the bound claimed in \cite{icalp11} is not achievable by that construction, observe that an adversary can put all their available false negatives and/or false positives in the list-disjunct matrix in one of the leaves of the recursion tree, totally sabotaging the decoding procedure. 

\subsection{Resolving the ``For-Each'' Case of Group Testing}

\begin{theorem} \label{thm:group_testing_for_each}
There exists a randomized construction of a matrix $\{0,1\} \in \mathbb{R}^{m \times n}$ with $m = O(k \log n)$ such that the following holds. Given $y = M \odot x$ with $|\mathrm{supp}(x)| \leq k$, we can find $x$ in time $O(k \log n)$, with failure probability $e^{-\Omega(k)} +  \frac{1}{\poly(n)}$.  
\end{theorem}

This theorem improves upon the recent work of \cite{bondorf2019sublinear}, which achieved the same number of rows but required quadratic running time in $k$. Our result essentially settles the non-uniform case of the group testing problem.

\subsection{$\ell_2/\ell_2$ Compressed Sensing}\label{subsec:ell2}

One of the central problems in compressed sensing is the design of an $\ell_2/\ell_2$ scheme, which is a matrix $\Phi \in \mathbb{R}^{m \times n}$, such that given $y = \Phi x$ we can find $x'$ satisfying 
	\[	\mathbb{P}\left \{	\|x-x'\|_2^2 \leq (1+\epsilon) \mathrm{min}_{\mathrm{k-sparse~}z} \|x-z\|_2^2	\right\} \geq 1- \delta.	\]

The goal is to randomly design $\Phi$ satisfying the above with the optimal number of rows, and enabling computation of such an $x'$ in sublinear-time (it can be proved that it suffices to pick $x'$ to be $O(k)$-sparse). Almost all sublinear-time algorithms (precisely, all but \cite{NakosS19}) proceed by reducing the problem to the construction of an $\ell_2/\ell_2$ weak identification)system\footnote{The authors in~\cite{ps12,glps17} define it in a slightly different way, and use slightly different terminology at places, but the essence of the property they demand is the same.}. This is a matrix $\Psi \in \mathbb{R}^{m \times n}$ such that given $y= \Psi x$ we can find $x'$ satisfying $\|(x-x')_{-k/2}\|_2 \leq (1+\epsilon) \|x_{-k}\|_2$ with probability $1-\delta$; recalling that $x_{-k}$ is the vector that occurs after zeroing out the $k$ largest in magnitude coordinates in $x$. For yet another intriguing consequence of our techniques, we give the strongest weak identification $\ell_2/\ell_2$ system available in the literature. On how that translates to $\ell_2/\ell_2$ schemes, we refer the reader to Section~\ref{sec:ell2}.

\begin{theorem}\label{thm:ell2}
There exists a randomized construction of an $\ell_2/\ell_2$ weak identification system with \[m = O\left((k/\epsilon) \log(n/k) + \frac{1}{\epsilon} \cdot \frac{\log (n/k)}{ \log \log (n/k)}\cdot \log(1/\delta)\right),\] which allows finding the desired $x'$ in time $O(m \log^2 m)$.
\end{theorem}
A comparison with previous work follows.
\begin{itemize}

\item The construction in \cite{glps12} requires $m=\Theta((k/\epsilon) \log(n/k)\cdot \log(1/\delta))$.

\item The construction in \cite{lnw17} achieves $m=O((k/\epsilon) \log(n/k) + \epsilon^{-1}\log(n/k)\log(1/\delta))$, but in order to run in near-linear time in $m$ storing an additional inversion table of size $\Omega(n)$ is required.

\item The strongest result in \cite{ps12,gnprs13} obtains \[m = O(\epsilon^{-4} k \log(n/k) (\log_k n)^\alpha + \epsilon^{-1}\poly(\log n)\log (1/\delta))\] and decoding time $\Omega( (k/\epsilon)^{2^{1/\alpha}}) \poly(\log n))$, for any $a<1$. The main source of sub-optimality is the invocation of a list-recoverable code based on the Loomis-Whitney inequality~\cite{nprr18}.
\end{itemize}

To the best of our knowledge, our work is the first to construct a near-optimal weak system with near-optimal decoding time (without using an additional inversion table as in~\cite{lnw17}). In fact, we are able to obtain stronger results for the general $\ell_2/\ell_2$ problem, but the argument is very lengthy and somewhat outside the scope of the technical contribution of this paper, so we decided to leave the most general result for a future publication.

\section{Overview of Techniques and Comparison with Previous Work}\label{sec:tech}

We first note that it is easy to \emph{augment} a disjunct matrix to provide
sublinear time recovery at cost of a logarithmic factor increase in the number
of tests. This can be done via a simple \emph{bit-masking} scheme, as 
for example used in \cite{ref:CR19b,lee2019saffron,ref:BCSYZ19} (a similar
method has been used in the context of compressed sensing and heavy hitters in 
\cite{ref:GSTV06, ref:BGIKS08, ref:CI17}).
That is, via the best known explicit constructions, to obtain a disjunct matrix with $O(k^2 \log^2 n)$ rows and similar decoding time, and a list-disjunct matrix with $O(k \log^2(n/k))$ rows and similar decoding time. 
On the other hand, much effort in the sparse recovery literature have been dedicated to simultaneously obtaining the ``ultimate goal'' of optimal number of measurements and near-linear decoding time \cite{ch09,gi10,indyk2010efficiently, icalp11, bi11,glps12, gnprs13, lnnt16, nakos2017deterministic_heavy_hitters,glps17,cai2017efficient,lee2019saffron,bondorf2019sublinear}, including work on the related Sparse Fourier Transform problem \cite{hikp12a,hikp12b,ik14,k16,k17,cksz17,kvz19}. In general, attaining the ultimate goal requires novel approaches and sophisticated techniques.

The approach of \cite{ch2009noise} lies in connecting disjunct and list-disjunct matrices (and related objects that suffice for the group testing problem) to randomness extractors and bipartite expanders, and then using the nice list-decoding properties of specific instantiations of those objects. The approach of \cite{indyk2010efficiently} is again closely related to list-recoverable codes. In a nutshell, the authors consider a variant of code concatenation, where the inner code is a random code that gives a list-disjunct matrix, and the outer code is a Reed-Solomon code. The concatenated code enables list-recovery, returning a list of size $O(k^2)$, which can then be filtered out using Lemma~\ref{lem:point_query_list_disjunct}. One of the crucial observations of that work is that it is sufficient to use an inner code which forms a list-disjunct matrix instead of a disjunct matrix, an approach that would require $\Omega(k^3 \log n/ \log k)$ rows. 

The first approach of \cite{icalp11} for list-disjunct is again a reduction to list-recoverable codes, making use of the list-decoding view of Parvaresh-Vardy codes. The second approach is to apply recursively the trivial list-recoverable code of block size $2$. In particular, they group together coordinates that agree in the first $(\log n)/2$ bits, and use a list-disjunct matrix in this instance to obtain a list $L_1$. Second, they group together coordinates that agree in the second $(\log n)/2$ bits and obtain a list $L_2$ using, again, a list-disjunct matrix. Thus, they can guarantee that the set of defective items lies in the set $L_1 \times L_2$, which is much smaller than $n$. Applying the same idea recursively and observing that the universe shrinks by a square root factor in each recursive call, yields in total a running time of $O(k^2 \poly(\log n))$ and a slight sub-optimality in the number of rows. This technique has also found use in subsequent works, such as \cite{ps12,gnprs13,nakos2017deterministic_heavy_hitters,inan2019optimality}. Instead of the trivial list-recoverable code, a slightly more efficient one based on the Loomis-Whitney inequality~\cite{gnprs13,nprr18} can be used, but this leads to a much less clean trade-off between the number of rows and the decoding time, and the total scheme is still sub-optimal in both measurement and time complexity.

There are many other papers that (possibly with minor modifications) implicit or explicitly construct efficiently decodable list-disjunct matrices, \cite{cai2013grotesque,cai2017efficient,lee2019saffron,bondorf2019sublinear} to name a few. However, all of those approaches fall short of bypassing the barriers mentioned in the first paragraph of this section and thus obtaining the ``ultimate goal'' of group testing.

Our approach departs from previous work and is inspired by the solution for the case $k=1$. First of all, we observe that it suffices to solve a slightly weaker-- but as it turns out crucially easier-- problem: find a list of size $O(k \log (n/k))$ that contains all defective items, or in other words construct a $(k, O(k \log(n/k)))$-list disjunct matrix, the decoding routine of which appears to be linear in the number of rows, as we show below. In particular, for every $\ell \in \{\log k,\ldots, \log n\}$, we group every coordinate which agrees in the first $\ell$ bits into a single coordinate, and then hash the single coordinate to $O(k)$ buckets. This yields in total $O(k \log (n/k))$ samples. Our algorithm then starts from $\ell = \log k$ and processes prefixes in increasing length $\ell$, trying to gradually find the prefixes of all the defective items, by discarding prefixes that do not participate in a positive test. Our proof roughly proceeds by showing that the number of possible trajectories of the non-defective items can be upper-bounded by the Catalan number of order $\Theta(k \log (n/k))$. The hashing scheme then allows for a union-bound over all possible supports and all possible trajectories. This immediately gives that the output list might contain up to $O(k \log (n/k))$ coordinates, which completes the proof. Our argument can facilitate up to $O(k)$ false positives, or even up to $O(k \log (n/k))$ if they are not ``very'' adversarially chosen; the latter guarantee (more than) suffices for our application to the heavy hitters problem. 

For the error-correcting schemes and the $\ell_2/\ell_2$ weak system, we need delicate twists in the hashing scheme and more involved analyses which make use of the generalized Catalan numbers to bound the running time. Unfortunately for $\ell_2/\ell_2$ compressed sensing the argument is not as clean (though the algorithm is still quite compact), but happens to be quite subtle on the technical level, mostly because one needs to handle dependent events carefully. Our approach is distinct from every previous work on the topic in several ways. Lastly, we also provide a novel reduction from near-linear decodable error-correcting disjunct matrices to ``less'' efficiently decodable error-correcting list-disjunct matrices, which does not require list-recoverability as the reductions in \cite{ch09,indyk2010efficiently}. 

From a coding theoretic perspective, one can view that the key to our progress is bypassing the ``for-all'' demand of list-recoverable codes: list-recovery ensures that for all choices of lists corresponding to symbols of the codeword, some ``desirable'' condition holds, namely that the set of possible codewords agreeing with all symbol-lists is small. In our case, we show that this is not needed, since the lists not only depend on the hidden set of defective items, but also are formed in a serial fashion, while one learns the bits of the defective items in chunks of appropriate size. 

\vspace{-4mm}
\paragraph{Connection to Tree Codes.} The above discussion might bear similarities to the constructions of tree codes \cite{treecodes91}. For alphabets $\Sigma, \Gamma$ and parameters $s,\delta$, a (truncated) tree code is a function $T\colon \bigcup_{r=1}^s \Sigma^r \rightarrow \bigcup_{r=1}^s \Gamma^r$ such that the Hamming distance of $T(x,x')$ for every two $x,x' \in \bigcup_{j=1}^s \Sigma^j$ is at least  $\delta \cdot (s- \mathrm{split}(x,x'))$, where $\mathrm{split}(x,x')$ is the largest index $i$ for which $x_j = x_j'$ for all $j \leq i$. In our case, one could imagine setting $s=\log n, \Sigma = \{0,1\}$ and $\Gamma = \{0,1\}^{\log k + O(1)}$ to pass to an efficiently decodable $k$-disjunct matrix. However, it turns out that the structural condition demanded by tree codes is too strong (and rather inflexible when trying to perform a disjunct or list-disjunct matrix construction) for group testing applications, as tree codes demand a worst-case Hamming distance bound between any two strings $x,x'$. On the other hand, what our approach requires is a Hamming distance bound that holds \emph{on~average}.
\vspace{-4mm}
\paragraph{Sublinear-Time Sparse Recovery Frameworks.}

In the sparse recovery literature, there are in principle two available frameworks for sublinear-time decoding. The first includes, as mentioned in the first paragraph of this section, bit-masking with an error-correcting code. This approach is particularly effective in compressed sensing tasks (such as $\ell_2/\ell_2$) where it is not necessary to detect every heavy coordinate of the vector $x\in \mathbb{R}^n$. In that case, one can recover a constant fraction of the coordinates and then exploit the linearity of the sketch in order to set up a clean-up process. This can be achieved by subtracting from $x$ the detected coordinates~\cite{bi11,glps12,hikp12a,hikp12b,gnprs13,k16,k17,lnw17,cksz17,ref:CI17}, in order to recover a heavy enough subset of the top $k$ coordinates rather than all of them. This bit-masking plus clean-up process is less effective for heavy hitters because one desires to recover \emph{all~of~them} rather than a constant fraction. It is not applicable to group testing, for subtraction is not possible in that model. Moreover, bit-masking always results in measurement-\emph{sub-optimal} schemes in the ``for-all'' case in every variation of the sparse recovery problem.

The other available framework is based on list-recoverable codes and related techniques \cite{ch2009noise,indyk2010efficiently,icalp11,gnprs13,lnnt16,glps17,nakos2017deterministic_heavy_hitters}. This type of machinery is more powerful in the ``for-all'' model, yielding better bounds in terms of measurement complexity, but it usually comes with a  polynomial in $k$ decoding time and complicated algorithms.

Our approach adds one more framework to the sparse recovery toolkit, and we demonstrate its power by obtaining a sequence of new and essentially optimal results that were not possible using previous arguments. We believe that further progress in the field could stem from hybrid approaches that combine more than one framework.

\section{Construction of the Identification Matrix}\label{sec:construction}

In this section, we present our (Monte-Carlo) construction of the identification matrix $M \in \{0,1\}^{m \times n}$. This matrix has $m \leq C k \log(n/k)$ rows and $n$ columns. As has already been stated, it is the vertical concatenation of $\log(n/k)$ matrices 

\[	M^{(\log k)}, M^{(\log k + 1)}, \ldots, M^{(\log n)},	\]
 each  consisting of $Ck $ rows.
As also stated in the preliminaries (Section~\ref{sec:prelim}), we will pick $n$ to be a power of two and identify $[n]$ with $\{0,1\}^{\log n}$ via the obvious bijection. 

The construction appears in Algorithm~\ref{alg:construction}. In each matrix $M^{(\ell)}$, all columns that agree in the first $\ell$ bits of their binary representation will have a $1$ in the same row. Then a random function $h_\ell\colon \{0,1\}^\ell \rightarrow \{0, \ldots, C k-1\}$ is chosen\footnote{It suffices to choose $h_\ell$ to be $O(k\log(n/k))$-wise independent.}, which maps (groups of) elements to rows.

One can view each $M^{(\ell)}$ matrix as a hashing scheme: all $i \in \{0,1\}^{\log n}$ are first grouped via the prefixes $\bpref_{\ell}$ to $2^{\ell}$ values, which are in turn hashed to $O( k )$ buckets; this means that if two $i, i' \in \{0,1\}^{\log n}$ agree in their first $\ell$ bits, then they will necessarily contribute to the same measurement in $M^{(\ell)}$.

\begin{algorithm}[!ht]
\begin{algorithmic}[1]\caption{Construction of the identification matrix $M$}\label{alg:construction}
\Procedure{\textsc{CreateMatrix}}{$k$}
	\For { $\ell = \log k$ to $\log n$}
		\State Initialize $M^{(\ell)}$ to have $0$ entries.
		\State Pick hash function $h_\ell\colon \{0,1\}^\ell \rightarrow \{0,\ldots,C k-1\}$.	
		\For { $i \in \{0,1\}^{\log n}$ }.
			\State $ q \leftarrow h_\ell(\bpref_{\ell}(i))$
			\State $M^{(\ell)}_{ q, i } = 1$
		\EndFor		
	\EndFor
	\State Return $M$ as the vertical concatenation of $M^{(\log k)}, M^{(\log k+1)}, \ldots, M^{(\log n)}$.
\EndProcedure
\end{algorithmic}
\end{algorithm}

\section{Decoding Algorithm Analysis and Proof of Theorem~\ref{thm:main_tool}}
In this section, we first give an analysis of the decoding algorithm from Theorem~\ref{thm:main_tool}. The algorithm is presented in Algorithm~\ref{alg:decoding}. In what follows, we will not explicitly use the definition of list-disjunct matrices the way they are stated, but we will argue the desired guarantees of our identification matrix ad-hoc, from first principles. 

In the following, we define $B = C \cdot k$. We also remind the reader that $M$ consists of the vertical concatenation of $M^{(\log k)}, \ldots, M^{(\log n)}$. We will say that a prefix $p$ of length $\ell$ contains an item $i$ if $\bpref_\ell(i) = p$. We assume that we work in a machine with word size $\Omega(\log n)$, such that indexing and concatenating strings of length $\Theta(\log n)$ takes $O(1)$ time.

\begin{algorithm}[!ht]
\caption{Given $y=M \odot x$, finds a list $L$ of size $O(k \log (n/k))$ that contains all defective items}
\label{alg:decoding}
\begin{algorithmic}[1]\Procedure{\textsc{Identify}}{$y$}
	\State $L =  \{0,1\}^{\log k} $
	\For { $\ell = \log k$ to $\log n$}
		\For {$p \in L$}
				\State $q \leftarrow  h_{\ell}(p)$ \Comment{Find in which row of $M^{(\ell)}$ the elements with prefix $p$ are set to $1$.}
				\State $z \leftarrow M_q ^{(\ell)} \odot x$ \Comment{Fetch the corresponding entry (by reading $y$).} \label{lin:z} 
				\If { $z = 0 $ } 
					\State Discard $p$ from $L$. \Comment{That entry should be $1$ if $p$ is the prefix of a defective.}	\label{lin:discard}
				\EndIf
		\EndFor
		\If {$\ell = \log n$}
			\State \Return $L$
		\EndIf
		\For {$p \in L$}
			\State Add $p \| 0$ and $ p \| 1$ to $L$.		\Comment{Extend the set of available prefixes in $L$.}
			\State Discard $p$ from $L$.
		\EndFor
	\EndFor
\EndProcedure
\end{algorithmic}
\end{algorithm}

We are now ready to proceed with the proof of Theorem~\ref{thm:main_tool}. For that, we state two lemmas.
The first lemma proves the easy fact that no defective will be left out of the output list $L$.

\begin{lemma}\label{lem:defective}
If $i$ is defective, then at the end of the execution of Algorithm~\ref{alg:decoding}, $i \in L$.
\end{lemma}

\begin{proof}
For $\ell \in \{\log k, \ldots, \log n\}$, the value of $z$ in Line~\ref{lin:z} depends on $h_\ell(\bpref_\ell(i))$ and will always be $1$. Hence, in the test in the next line, 
the prefix $p$ of $i$ will not be discarded from the list. In the end, when $\ell = \log n$, we conclude that $i$ will be in $L$.\end{proof}

What remains is to bound the size of the list $L$ that our algorithm outputs. We will show that with probability $1 - e^{-C_1 k\log(n/k)}$ the list will always have $\CL k\log(n/k)$ items in it, where $C_1$ is an absolute constant. 

\begin{lemma}\label{lem:small_list}

For any constant $\CFP,C_1$ there exist (constants) $C,\CL$ such that

\[		\mathbb{P}\left[ \exists x \in \{0,1\}^n,|\mathrm{supp}(x)| \leq k\colon |\textsc{Identify}(M \odot x)| > \CL k \log(n/k) + k\right] < e^{-C_1 k \log(n/k)}, \]
under the presence of $\CFP k$ adversarial false positives per matrix $M^{(\ell)}$. The randomness is over the functions $\{h_\ell \}_{\ell \in \{\log k,\ldots, \log n\}}$. Moreover, the running time of Algorithm~\ref{alg:decoding} is $O(k \log (n/k))$.
\end{lemma}

To prove Lemma~\ref{lem:small_list}, we note that the execution of our decoding algorithm produces a binary forest, consisting of $k$ trees rooted at level $\log k$. First of all, define the binary tree $\mathcal{T}$ of depth $\log n$ with nodes indexed by binary strings (prefixes). The tree $\mathcal{T}$ is rooted at the empty string; taking a path to the left appends $0$ to the current string, otherwise, it appends $1$. A string $p$ of length $\ell$ has children $p \| 0$ and $p\| 1$ of length $\ell+1$. Moreover, the relation $T \subseteq_{\mathrm{tree}} \mathcal{T}$ will denote the fact that $T$ is a connected sub-tree of $\mathcal{T}$. 

Given the above definitions, we can think of our algorithm as performing the following natural procedure: It starts with all strings of length $\log k$, maintaining a list $L$ of possible prefixes. In the beginning, all $k$ possible binary strings are in $L$. For each length $\ell$, the algorithms considers every $p \in L$, and checks whether it participates in a negative test. If this is not the case, $p$ is replaced by the two prefixes that can be obtained by appending a $0$ or a $1$ to it. The prefix $p$ is always discarded for the next iteration
For a given $k$-sparse vector $x \in \{0,1\}^n$, we define the trajectory $\tau_x$ of the items to be the set of all possible prefixes $p$ of $x$ that might be inserted in $L$ at some point during the execution of the algorithm and are not discarded in Line~\ref{lin:discard}. We also refer to the trajectory of the defective items, and denote it, for a vector $x \in \{0,1\}^{n}$, as $\tau_x'$, as the set of possible prefixes $p$ containing a defective item which will be inserted in $L$ at some point during the execution of the algorithm, for some $\ell \in \{\log k, \ldots, \log n\}$.
It should be clear that $\tau_x$ (and $\tau_x'$) is a forest consisting of $k$ trees rooted at level $\log k$, and with at least $|\mathrm{supp}(x)|$ leaves because of Lemma~\ref{lem:defective}. The proof proceeds by showing that for any $x \in \{0,1\}^n$ with $|\mathrm{supp}(x)| \leq k$ it holds that $|\tau_x| + |\tau_x'| = O(k \log (n/k))$ with probability $1 - e^{-\Omega(k\log(n/k))}$, which then implies Lemma~\ref{lem:small_list}.

\begin{proof}[Proof of Lemma~\ref{lem:small_list}]
Define the forest $f_x = \tau_x \setminus \tau_x'$; i.e., the trajectory of the non-defective items once they are separated from the defective items. It suffices to prove that $|f_x| \leq \CL k \log (n/k)$, since $L$ at the end can contain at most $k + |f_x| \leq k + \CL k \log(n/k)$ elements. Similarly, the prefixes inserted in $L$ are in total $|\tau_x| + 2|\tau_x'| \in O( k\log(n/k))$, hence the bound on the running time. 

Fix $E^{(\log k)}, \ldots, E^{(\log n)} \subseteq [C k]$ of size at most $\CFP \cdot k$, which correspond to the false positives in each of the matrices $M^{(\log k)}, \ldots, M^{(\log n)}$. For a prefix $p$ of length $\ell \geq \log k$, define a binary random variable $Y_{p}$, such that $Y_{p} = 0$ if and only if 

	\[	\left( (M^{(\ell)} )_{h_{\ell}(p)}  \cap  \mathrm{supp}(x) = \emptyset \right) ~\mathrm{and}~ \left( h_\ell(p) \notin E^{(\ell)}\right). \]
 Observe that for $p$ not containing a defective item we have that

\[	\mathbb{P} \left[ Y_{p}  = 1\right] \leq \frac{k}{Ck} + \frac{\CFP k}{Ck} = \frac{1 + \CFP}{C}.	\]
In words, $Y_p$ captures an event where $p$ participates in a positive test
(i.e., either by the existence of a defective or a false positive). For $p$ of length smaller than $\log k$ we deterministically set $Y_p=1$.
 We now have that 

\begin{align}
\mathbb{P} \left[ |f_x | > \CL k \log(n/k) \right]
 &\leq \mathbb{P} \left[ \exists T \subseteq_{\mathrm{tree}} \mathcal{T}, |T| = (\CL +1)k\log(n/k) +k\colon \forall p \in T, Y_p =1 	\right] \nonumber \\
&\leq \mathrm{Cat}^2_{(\CL+1)k\log(n/k) + k} \left( \frac{1 + C_{\mathrm{FP}}}{C} \right)^{\CL k \log(n/k)} \label{pr:b}\\
&\leq 4^{(\CL+1)k \log (n/k) + k} \left( \frac{1 + C_{\mathrm{FP}}}{C} \right)^{\CL k \log(n/k)} \label{pr:c} \\
&\leq \left ( \frac{4^3 ( 1+\CFP)}{C} \right)^{\CL k \log (n/k)}.
\label{pr:d}
\end{align}
 
In the above, \eqref{pr:b} follows by the observation that $f_x$ can be extended to a binary tree rooted at the empty prefix 
 by adding $k + k \log(n/k)$ additional nodes, \eqref{pr:c} follows by a union bound over all binary trees  with $(\CL+1)k \log(n/k) + k$ internal nodes, and \eqref{pr:d} by Lemma~\ref{lem:estimation_catalan}. 
We now have that 

\begin{align*}
&\mathbb{P} \left[ 
\exists (x \in \{0,1\}^n,\{E^{(\ell)}\} \subseteq [Ck])\colon |\mathrm{supp}(x)| \leq k, (\forall \ell) |E^{(\ell)}| \leq \CFP k, |f_x| > (\CL+1)k \log(n/k) \right] \\
&\leq \underbrace{\left( \sum_{j=0}^k {n \choose k} \right)}_{\mathrm{choices~for~}x} \cdot \underbrace{{Ck \choose k}^{ \log(n/k)}}_{\mathrm{choices~for~}E^{(\ell)}} \cdot \left ( \frac{4^3 ( 1+\CFP)}{C} \right)^{\CL k \log (n/k)} \\
&\leq (k+1) \cdot {n \choose k} \cdot (e C)^{k \log (n/k)} \cdot \left ( \frac{4^3 ( 1+\CFP)}{C} \right)^{\CL k \log (n/k)} \\
&\leq (k+1) \cdot e^{k \log (en/k)} \cdot (eC)^{k \log (n/k)} \cdot  \left( \frac{4^3 ( 1+\CFP)}{C}\right)^{\CL  k \log (n/k)},
\end{align*}
where we have used the fact that ${n \choose k}$ is increasing for $0 \leq k \leq n/2$ and the standard inequality ${ a \choose b} \leq (ae /b)^b$ for integers $a,b$. By choosing $ C \geq 2 \cdot ( 4^3 ( 1+\CFP))$ and $\CL$ such that $2^{\CL} \geq   e^{3+C_1} C $, we obtain that the latter bound is at most $e^{-C_1 k \log (n/k)}$.
This completes the proof of the lemma.

\end{proof}

We are now ready to prove Theorem~\ref{thm:main_tool}.
\begin{proof}[Proof of Theorem~\ref{thm:main_tool}]
The proof, apart from the last sentence of the statement, follows immediately by combining Lemmas \ref{lem:small_list}~and~\ref{lem:defective}. For the other part, since as mentioned every $h_\ell$ is a $O(k \log(n/k))$-wise independent hash function, storing all hash functions in a straightforward way would require $\Omega(k \log^2(n/k))$ words of space, which is prohibitive. However, we may observe that the $h_\ell$ have the same range and domains of exponentially increasing size, so we can pack them in a single hash function. Let $g\colon \{0,1\}^{\log n + 1} \rightarrow \{0, \ldots, Ck-1\}$ be a random $((\CL+1) k\log(n/k))$-wise independent hash function. For $\log k \leq \ell \leq \log n$ define $h_\ell$ as the restriction of $g$ on the binary strings the value of which in the decimal system is in the set $\{2^\ell,2^\ell+1,\ldots,2^{\ell+1}-1\}$. Now, all we need for Lemma~\ref{lem:small_list} to go through is that every $(\CL+1) k\log(n/k)$ points are independently mapped to $\{0,1,\ldots,Ck-1\}$ under application of potentially different $h_\ell$; this is guaranteed by the $((\CL+1) k \log(n/k))$-wise independence of $g$. Furthermore, we shall use as $g$ the standard construction of $\kappa$-wise independent hash functions with $\kappa$ degree polynomials, for $\kappa = (\CL+1) k\log (n/k)$. Fast multi-point evaluation of polynomials allows evaluating a polynomial of degree $\kappa-1$ in $\kappa$ points in time $O( \kappa \log^2 \kappa \log \log \kappa)$ in the word RAM model, \cite[Corollary 10.8]{algebra}. Splitting $i_1,i_2,\ldots,i_B$ into batches of size $\kappa = (\CL+1) k \log (n/k)$ and adding dummy points if needed, the multi-point evaluation of $g$ gives multi-point evaluation for each $h_\ell$, yielding the desired result.

\end{proof}

\section{Obtaining Optimal Monte Carlo Constructions}
\label{sec:obtain}

In this section, we show how to use Theorem~\ref{thm:main_tool} to obtain Theorems \ref{thm:list_disjunct},~\ref{thm:disjunct},~\ref{thm:heavy_hitters},~\ref{thm:heavy_hitters_estimates}, and~\ref{thm:group_testing_for_each}.

\begin{proof}[Proof of Theorem~\ref{thm:list_disjunct}]
We augment the matrix guaranteed by Theorem~\ref{thm:main_tool} with the matrix guaranteed by Theorem~\ref{thm:list_disjunct_tool}. By the first matrix, we can find a list $L$ of size $O(k \log (n/k))$, which in turn can be filtered out by the decoding algorithm in time $O(\log (n/k))$ per element, using Lemma~\ref{lem:point_query_list_disjunct}.
\end{proof}

\begin{proof}[Proof of Theorem~\ref{thm:disjunct}]
We augment the matrix guaranteed by Theorem~\ref{thm:list_disjunct} with the matrix guaranteed by Theorem~\ref{thm:disjunct_tool}. By the first matrix, we can find a list $L$ of size $2 k $, which in turn can be filtered out by the decoding algorithm in time $O(k \log_k n)$ per element, using Lemma~\ref{lem:point_query_disjunct}. The total number of rows is $O(k\log(n/k) + k^2\mathrm{min}\{\log n, ( \log_k n)^2\}) = O(k^2\mathrm{min}\{\log n, (\log_k n)^2\}) $ and the running time is $O(k \log^2(n/k) + m)$. In particular, performing point-queries using part i) of Theorem~\ref{thm:disjunct_tool} can be done in time $O(k) \cdot O(k \log n) = O(k^2 \log n)$, whereas using part ii) of Theorem~\ref{thm:disjunct_tool} can be done in time $O(k) \cdot O(k\log_k n) = O(k^2 \log_k n)$. In both cases, we shall obtain the claimed running time.
\end{proof}

\begin{proof}[Proof of Theorem~\ref{thm:heavy_hitters}]
Let us first prove the theorem in the case where we can store the whole matrix $M$, and afterward turn the obtained scheme to a streaming algorithm with the desired guarantees.

We augment the matrix $M$ guaranteed by Theorem~\ref{thm:main_tool} with the matrix guaranteed by \cite{lnw17}, along with a single-row matrix consisting of the all $1$s vector. The third matrix gives us $\|x\|_1$. The second matrix has $O( \log (n/k))$ non-zeros per column. Moreover, it allows filtering out (via point-queries) any list $L$ of $i \in [n]$ in time $|L| \log(n/k)$, similarly to list-disjunct matrices, such that we are left with a list of size $O(k)$ that contains all $i \in L$ with $x_i \geq (1/k) \|x\|_1$. Using the first matrix, we shall show that given $y = Mx$ with $x \in \mathbb{R}_+^n$ we can find a list of size $O(k\log(n/k))$ that contains every $(1/k)$-heavy hitter. Combining with the second matrix we shall obtain the desired result. We set each measurement (bucket) to $1$ if $y_q \geq \|x\|_1/k$, and $0$ otherwise; we can implement this test since we know $\|x_1\|$ exactly. Thus, we obtain a $y \in \{0,1\}^{C k\log(n/k)}$, on which we run the group testing procedure guaranteed by Theorem~\ref{thm:main_tool}. Note that if a $(1/k)$-heavy hitter participates in a measurement $q$ then $y_q=1$. Otherwise, in each of the sub-matrices $M^{(\log k)},\ldots, M^{(\log n)}$ there can be at most $k$ false positives since in each sub-matrix every $i \in [n]$ participates in exactly one measurement. The guarantee of Theorem~\ref{thm:main_tool} yields the desired result.

\newcommand{\buf}{\mathrm{buf}}
We now turn the above scheme to an efficient low-space data structure. We first pick the low-space representation of $M$ guaranteed by Theorem~\ref{thm:main_tool} and observe that for a fixed $\ell$ all the fetches in Line~\ref{lin:z} of Algorithm~\ref{alg:decoding} can be performed in time $O(k \log (n/k) \log^2(k \log (n/k)) \cdot \log \log (k \log (n/k)))$ by the last sentence of Theorem~\ref{thm:main_tool}. Since there are only $\log (n/k)$ levels, we get a $k \poly( \log n)$ decoding time. Implemented naively, the update time is $O(k \log^2 (n/k))$, as we have to evaluate $h_{\log k}, \ldots, h_{\log n}$, each being $O(k \log(n/k))$-wise independent.

To improve the update time, we invoke an argument from \cite{KaneNPW11,alman2020faster} (which will result in a non-linear sketch). First of all, we can keep a buffer of size $\Theta(k\log(n/k))$ which stores updates $(i,\Delta)$, performing all of them (and flushing the buffer) when it fills up or when a query comes. Note that an update is not performed upon arrival, but only when the buffer is full or upon a query. All the updates can be performed in time $O(k \log^2(n/k) \log^2 (k \log (n/k)) \log \log (k \log (n/k)))$, using the fast batch location of Theorem~\ref{thm:main_tool}. This yields an \emph{amortized} cost of $\tilde{O}(\log(n/k) \log^2k)$. We shall show how to \emph{de-amortize} this cost, for a worst-case cost of $\tilde{O}(\log(n/k) \log^2k)$. We shall keep two buffers $\buf_0, \buf_1$ of size $B$ in words, for $B = O(k \log (n/k))$. Each time the algorithm receives an update, it puts it to $\buf_0$. Once $\buf_0$ reaches its maximum size, the next update is put in $\buf_1$ and in parallel $\buf_0$ is flushed by performing all the updates for $O( k \log^2 (n/k) \cdot \log^2(k \log (n/k)) \log \log k / B ) = O( \log (n/k) \cdot \log^2(k \log (n/k)) \log \log k)$ steps using the fast batch location detection guaranteed by Theorem~\ref{thm:main_tool} (this means that we perform a certain number of operations and then pause, waiting for the next update, and continue the execution once the next update comes). Once $\buf_1$ becomes full, the roles of $\buf_0$ and $\buf_1$ are switched. Upon query, we may flush both buffers by using the fast batch location detection guaranteed from Theorem~\ref{thm:main_tool}. This de-amortizes the update time, spreading it over $O(k \log (n/k))$ steps, and increases the query time by an additive $O(k \log^2(n/k) \log^2 (k \log (n/k)) \log \log (k \log (n/k))) = O(k \poly( \log n))$ factor. We can similarly de-amortize the update time of the algorithm in \cite{lnw17}. Putting everything together gives the desired result.

\end{proof}

\begin{proof}(Theorem~\ref{thm:heavy_hitters_estimates})
It is proved in \cite{NNW12} that there exists a matrix with 
\[k^2 \mathrm{min}\left\{\log n, \left(\frac{\log n }{ \log \log n + \log k }\right)^2 \right\}\]
rows and $k \log n$ column sparsity, which allows deterministic $\ell_1$ point queries for the heavy hitter problem; i.e., find $x_i'$ such that $|x_i - x_i'|\leq (c/k) \|x\|_1$. Those constructions are the same as the standard constructions of $k$-disjunct matrices; i.e., via a random code (hashing to $\Theta(k)$ buckets, and repeating $\Theta(k \log n)$ times) or via Reed-Solomon codes with the alphabet and block size equal to $\Theta(k \log n / (\log k +\log \log n))$. Moreover, they have low-space representations, as for the former only constant-wise independent hash functions per repetitions are required, while the latter is strongly explicit. Moreover, we can perform a point query and an update using the first matrix in time $O(k \log n)$. For the Reed-Solomon construction, every column $i \in [n]$ corresponds to a polynomial of degree $O( \log n / ( \log \log n + \log k))$, and every repetition corresponds to a symbol of the codeword; thus we can perform updates and point-queries for $i \in [n]$ in time $O(k (\log n / (\log k + \log \log n)) ^2)$.

These matrices are $0,1$ matrices, and thus for the strict turnstile model we easily have that $x_i'$ satisfies $x_i \leq x_i' \leq x_i + (c/k) \|x\|_1$. Augmenting the sketch constructed in Theorem~\ref{thm:heavy_hitters} with the aforementioned sketch of \cite{NNW12}, the query algorithm finds the list $L$ guaranteed by Theorem~\ref{thm:heavy_hitters}, and then discards all $i \in L$ with $x_i' < (1/k)\|x\|_1$. This yields the desired result.
\end{proof}

We note that using our approach we can also obtain a streaming algorithm with $s = O(k^2 \log^2 n)$ for $\ell_1$ heavy hitters in the turnstile model where we still care about all $i \in [n]$ with $|x_i| \geq (1/k) \|x\|_1$ and we do not have the assumption that $x_i \geq 0$; the query time is $O(m \log n)$. We sketch how. We pick $O(k \log n)$ permutations of $[n]$, and run in each permuted universe a variant of the decoding algorithm in Theorem~\ref{thm:heavy_hitters} (starting from level $\lceil \log (Ck) \rceil$ for some large enough constant $C$), where instead of implementing the test we keep the top $O(k\log n )$ coordinates with the largest estimates for every $\ell \in \{\lceil \log (Ck) \rceil, \ldots,\log n\}$. It can be proved that this algorithm returns a list of size $O(k \log n)$ that contains all $\ell_1$ heavy hitters; part of the argument is common to the one in Section~\ref{sec:ell2}. Using the sketches in \cite{NNW12} we can find point estimates for every $x_i,i \in L$, and we may keep the top $2k$ coordinates. This result is not novel since a scheme with the same guarantee can be obtained by bootstrapping the incoherent matrix of \cite{NNW12} with a linear-time decodable error-correcting code that can correct a constant fraction of errors~\cite{s96}. What is possibly novel is that error-correcting code machinery is not necessary in order to obtain $O(k^2 \log^2 n)$ running time. It is known how to achieve a smaller number of rows \cite{nakos2017deterministic_heavy_hitters}, but the dependence of the decoding time on $k$ was $k^7\poly(\log k)$ and the argument was significantly more complicated.

\begin{proof}(Theorem~\ref{thm:group_testing_for_each})
We augment the matrix $M$ guaranteed by Theorem~\ref{thm:main_tool} with two matrices $M',M''$. The matrix $M'$ has $O(k \log (n/k))$ rows and $n$ columns, with each column having exactly one $1$ in a random position. The matrix $M''$ has $O(k \log n)$ rows and it is the vertical concatenation of $O( \log n)$ matrices, each one with $O(k )$ rows: each column in each sub-matrix has exactly one $1$ at a random position. As before, using $M$ we obtain $L$ of size $O(k \log (n/k))$ which contains all defective items. Using $M'$ we can next filter out the non-defective items from $L$ to obtain a list of size $O(k)$ in $O(k \log (n/k))$ time using standard point-queries. The probability that a non-defective item participates in the same measurement with a defective item in $M'$ is $O(1/ \log(n/k))$. Thus, by an application of the additive form of the Chernoff bound, we shall discard all but $O(k)$ non-defective items in $L$, with probability $1 - e^{-\Omega(k)}$. Filtering out the rest of the items using $M''$ to obtain the defective items can be done in $O(k \log n)$ time; the proof of correctness is standard.

\end{proof}

\section{An Explicit Construction of Efficiently Decodable $k$-Disjunct Matrices} \label{sec:explicit}

It is natural to ask whether our construction can be derandomized in order to obtain explicit list-disjunct and disjunct matrices. For $k = O(1)$, for example, there are polynomially many events involved in the analysis of Theorem~\ref{thm:main_tool}, and we can use the method of conditional expectations to derandomize the hashing scheme in polynomial time. For large $k$, it could be possible to perform the method of conditional expectations by exploiting the tree-like structure of the events. However, we note that it seems getting list-disjunct matrices with the optimal number of rows may be out of reach at the moment, since the notion is closely connected with explicit constructions of extractors and unbalanced bipartite expanders (more generally, \cite{ch2009noise} allows the whole spectrum of ``condenser graphs'' ranging from lossless expanders to extractor graphs). Indeed, all explicit constructions of list-disjunct matrices that we are aware of follow essentially by such expander graphs, and the notion of list-disjunctness is essentially a notion of expansion on sets of size $k$.
In fact, this can be made rigorous at least for certain natural cases. When the decoder algorithm is the naive decoder that performs point queries, this can be seen as the list-decoding view of expander graphs (where the test matrix is the adjacency matrix of a left-regular bipartite graph) as described by Vadhan\footnote{%
Vadhan's characterization \cite[Proposition~7]{ref:Vad10} is qualitatively the following: 
For an unbalanced bipartite graph,
and a right vertex $T$, let $\mathrm{LIST}(T)$ denote the set of left-vertices whose
neighbors all fall in $T$. Then, the graph is an expander iff for every $T$, the set
$\mathrm{LIST}(T)$ is small. On the other hand, the naive decoder for list disjunct matrices,
when seen as a graph, precisely computes the set $\mathrm{LIST}(T)$ from the test outcomes
$T$, and list disjunctness ensures that this list size is small.
supported on $T$, and list disjunctness ensures that this list size is small.
It is worthwhile to mention that for the related notion of binary compressed sensing matrices, an equivalence with expander graphs has been proved in \cite[Theorem~2]{ref:BGIKS08}.
} \cite[Section~5]{ref:Vad10}.

For disjunct matrices, however, explicit constructions are obtained via incoherent matrices where things are better understood. In particular, we know a polynomial-time construction of a $(1/k)$-incoherent matrix \cite{porat2008explicit} with $O(k^2 \log n)$ rows, and thus of a $k$-disjunct matrix. Trying to derandomize the scheme in Theorem~\ref{thm:main_tool}, we found a surprisingly simple construction of explicit $k$-disjunct matrices. It should still be the case, however, that one can use the additional $k$ factor allowed by $k$-disjunct matrices to derandomize a variant of Theorem~\ref{thm:main_tool}, by resorting to a bound over a smaller number of events. The challenge is to formalize the correct notion of Hamming distance \emph{on~average} after the split point, as also mentioned in the discussion on tree codes in Section~\ref{sec:tech}, if that's possible. 

We shall make use of the following two Theorems.

\begin{theorem}[\cite{indyk2010efficiently}] \label{thm:indyk_derandomization}
Let $k = O( \log n / \log \log n)$. There exists a polynomial-time construction of a $k$-disjunct matrix with $O(k^2 \log n)$ rows which allows decoding in time $\poly(k, \log n)$.
\end{theorem}

\begin{theorem}[Corollary 14 in~\cite{ch2009noise}] \label{thm:list_disjunct_extractor}
There exists a strongly explicit $(k,O(k))$-list-disjunct matrix with $O(k\cdot 2^{O((\log \log k)^2)}\log n)$ rows. The column sparsity is $2^{O((\log \log k)^2)} \log n$. 
\end{theorem}

We are now ready to proceed with the proof of Theorem~\ref{thm:deterministic_disjunct}. The idea is that for small $k$ there is a small number of events to perform the method of conditional expectations, while as $k$ gets larger we can construct a (sub-optimal) efficiently decodable list-disjunct matrix exploiting the additional multiplicative $k$ factor we have in our possession.

\begin{algorithm}[!ht]
\begin{algorithmic}[1]\caption{Construction of the matrix of Theorem~\ref{thm:deterministic_disjunct}, for $k=\Omega\left(\frac{\log n}{\log \log n}\right)^{1+o(1)}$}\label{alg:explicit_construction}
\Procedure{\textsc{CreateMatrix}}{$k$}
	\State $D \leftarrow  \log n / \log \log n$	\Comment{Round $D$ to a power of $2$.}
	\State $H = \log_D n$ \Comment{$H= \Theta(\log n / \log \log n)$}
	\State $d \leftarrow \log D$
	\For { $\ell = 1$ to $H$ }
		\State Pick a $(k,O(k))$-list-disjunct matrix $T^{(\ell)} \in \{0,1\}^{m_\ell \times 2^{d\ell}}$ via Theorem~\ref{thm:list_disjunct_extractor}. \label{lin:explicit_pick} 
		\For {$i \in \{0,1\}^{\log n}$}
			\For {$q \in [m_\ell]$}
			\State $M^{(\ell)}_{q,i} \leftarrow T^{(\ell)}_{q,\bpref_{\ell d}(i)}$\label{lin:explicit_extend}
			\EndFor
		\EndFor
	\EndFor
	\State Return $M$ as the vertical concatenation of $M^{(1)}, M^{(2)}, \ldots, M^{(H)}$.
\EndProcedure
\end{algorithmic}
\end{algorithm}

\begin{proof} (Theorem~\ref{thm:deterministic_disjunct})
For $k = O(\log n / \log \log n)$ Theorem~\ref{thm:indyk_derandomization} gives $O(k^2 \log n)$ rows and decoding time $\poly( k , \log n) = k^2 \poly(\log n)$. We shall focus on $k = \omega((\log n/ \log \log n)^{1+o(1)})$. 
The algorithm uses a similar hierarchical decomposition as in Theorem~\ref{thm:main_tool}, albeit with a different argument and functionality. Let $D = \log n / \log \log n $, rounded to a power of $2$, $d= \log D$, and $H = \Theta(D)$ be chosen such that $n \leq D^H =  \Theta(n)$. For each $\ell \in [H]$ we shall keep an appropriate disjunct matrix $M^{(\ell)}$ to guide the search of the defective items, in the following way. A $(k,O(k))$-list-disjunct matrix is initially constructed over the universe $\{0,1\}^{d \cdot \ell}$ (Line~\ref{lin:explicit_pick} in Algorithm~\ref{alg:explicit_construction}) using Theorem~\ref{thm:list_disjunct_extractor}, and then extended to a matrix $M^{(\ell)}$ over $\{0,1\}^{\log n}$, by grouping together coordinates $i \in \{0,1\}^{\log n}$ with the same value $\bpref_{d \ell} (i)$ (Line~\ref{lin:explicit_extend} in Algorithm~\ref{alg:explicit_construction}). The total number of rows is 
\[	\sum_{\ell=1}^{H}O\left(k\cdot 2^{O((\log \log k)^2)}\ell d\right)   \leq (\log n /\log \log n) \cdot O(k \cdot 2^{O((\log \log k)^2)} \log n)  = O(k^2 \log n),	\]
for our regime of interest. The decoding algorithm processes the prefixes similarly to Algorithm~\ref{alg:decoding}. The algorithm proceeds in iterations, by processing strings of length $\ell d$ in the $\ell$-th iteration. At all times, it maintains a list $L$ of size $O(k)$, such that the end $L$ contains all defective items. In the $\ell$-iteration, for $p \in L$ it performs point-query on $p$ using $M^{(\ell)}$ and Lemma~\ref{lem:point_query_list_disjunct} to decide whether to proceed and examine the $d$ elements $\bigcup_{p' \in \{0,1\}^d} p \| p'$. Correctness is immediate by the definition of list-disjunct matrices. The total running time is \[O(k) \cdot (\log n / \log \log n) \cdot k^{o(1)} \log n = O(k^2 \log n ).\] In the end of the algorithm, we will find a list $L$ of size $O(k)$ which contains every defective item. Using the standard explicit construction of $k$-disjunct matrices with $O(k \log n)$ column sparsity from \cite{porat2008explicit} we can find exactly the defective items by Lemma~\ref{lem:point_query_disjunct}.
\end{proof}

\vspace{-2mm}
\section{Error-Correcting Disjunct and List-Disjunct Matrices
(Proof of Theorems \ref{thm:error_list}~and~\ref{thm:error_dis})}\label{sec:error}
\newcommand{\cnt}{\mathsf{count}}

This section is dedicated to proving Theorems \ref{thm:error_list}~and~\ref{thm:error_dis}.  We will use the following result.


\begin{theorem}[Error-Correcting Disjunct and List-Disjunct Matrices \cite{ch2009noise,icalp11}] \label{thm:error_tool} There exist
randomized constructions of
\begin{enumerate}
\item A $(k,k,e_0,e_1)$-error-correcting list-disjunct matrix with $ m = O(k \log (n/k) + e_0 + k e_1)$ rows. The column sparsity is $O(\log (n/k) + e_0 / k  + e_1)$.
\item A $k$-disjunct matrix with $m = O(k^2 \log n + e_0 + ke_1)$ rows which can tolerate up to $e_0$ false positives and $e_1$ false negatives. The column sparsity is $O(k \log n + e_0/ k + e_1 )$.
\end{enumerate}
\end{theorem}

\subsection{Proof of Theorem~\ref{thm:error_list}}
\begin{algorithm}[!ht]
\caption{Construction of error correcting $M$}
\label{alg:error_construction}
\begin{algorithmic}[1]
\Procedure{\textsc{CreateMatrix}}{$k$}
	\State $R \leftarrow C\log k$
	\State $d \leftarrow \lceil \alpha \log k \rceil$
	\State $D \leftarrow 2^d$ 
	\State $H \leftarrow \log_D (n/k) + 1$
	\For { $\ell=0$ to $H-1$}
		\For { $r=0$ to $R-1$}
		\State Pick hash function $h_{\ell,r}\colon \{0,1\}^{\log k + \ell\cdot  d} \rightarrow \{0,\ldots,C k-1\}$.	
			\For { $i \in \{0,1\}^{\log n}$ }
				\State $ q \leftarrow h_{\ell,r}(\bpref_{\log k + \ell \cdot d }(i))$
				\State $M^{(\ell,r)}_{ q, i } = 1$
			\EndFor
		\EndFor	\Comment{Every uninitialized entry of $M^{(\ell,r)}$ is set to $0$.}
	\EndFor
	\State \Return $M$ as the vertical concatenation of $M^{(\ell,r)}, (\ell, r) \in [H]\times [R]$.
\EndProcedure
\end{algorithmic}
\end{algorithm}


\begin{algorithm}[!ht]
\caption{Decoding procedure for error-correcting list-disjunct matrices}
\label{alg:error_decoding}
\begin{algorithmic}[1]\Procedure{\textsc{Identify-under-Errors}}{$y$} 
	\State $R \leftarrow C\log k$
	\State $d \leftarrow \lceil \alpha \log k \rceil$
	\State $D \leftarrow 2^d$ 
	\State $H \leftarrow  \log_D (n/k)$
	\State $L \leftarrow \{0,1\}^{\log k} $
	\For { $\ell=0$ to $H$}
		\For {$p \in L$} 
			\State $\cnt_p \leftarrow 0$
		\EndFor
		\For {$p \in L$}
			\For { $r=0$ to $R-1$} 
				\State $q \leftarrow  h_{\ell,r}(p)$ \Comment{Find in which row of $M^{(\ell,r)}$ the elements with prefix $p$ are set to $1$.} 
				\State $z \leftarrow M_q^{(\ell,r)} \odot x$ \Comment{Fetch the corresponding entry (by reading $y$).}
				\If { $z = 0 $ } 
					\State $\cnt_p \leftarrow \cnt_p + 1$
				\EndIf
			\EndFor
		\EndFor
		\State Discard every $p \in L$ with $\cnt_p > R/2$.	\label{lin:error_discard}
		\If {$\ell = H-1$}
			\State \Return $L$.
		\EndIf
		\For {$p \in L$}
			\State Add $ \bigcup_{p' \in \{0,1\}^d} \{p \| p' \}$ to $L$	\Comment{Expand $L$.}
			\State Discard $p$ from $L$
		\EndFor
	\EndFor
\EndProcedure
\end{algorithmic}
\end{algorithm}

First of all, let us reduce to the case where $e_0 \leq k \log k$ (this reduction applies also to the case of presence of false negatives, as can be easily inferred). Indeed, to construct a $(n,k,e_0,0)$-list-disjunct matrix we may vertically concatenate $1 + \left\lceil \frac{e_0 }{k  \log k }\right \rceil$ error-correcting $(n,k,e_0',0)$-list-disjunct matrices with $e_0' \leq k \log k$. Then, in order to decode we run~\emph{in~parallel} the decoding procedure on all those matrices, and halt when the first stops, returning the list $L$ it returns. Note that there exists at least one matrix which receives at most $k \log k$ false positives, and its decoding procedure will run correctly as desired. Increasing the number of false positives can only increase the running time of the decoding algorithm we are going to present, and thus the matrix that finishes first will be the one with the smallest number of false positives, giving us the desired result. The total number of rows is 

\[	m = O(k \log(n/k) ) \cdot \left( 1 + \left\lceil \frac{e_0 }{k  \log k }\right \rceil\right) 	= O(k \log (n/k) + \log_k n \cdot e_0 ) .\] 
The running time is then
\[	O(k^{1+\alpha} \poly(\log n) \cdot  \left( 1 + \left\lceil \frac{e_0 }{k  \log k }\right \rceil\right) = O( k^\alpha \cdot m  \cdot \poly(\log n)).	\]

Let us now drop the notation on $e_0$, assuming that we have at most $k \log k$ false positives. We shall show how to construct a matrix $M$ with $O(k \log(n/k))$ rows associated with a decoding procedure that allows us to find a list of size $k^{1+\alpha}\poly(\log n)$ which contains every defective item; i.e., an analog of Theorem~\ref{thm:main_tool}. Then, the list can be filtered out at the same cost using the matrix guaranteed by Theorem~\ref{thm:error_tool}. 

The construction of $M$ appears in Algorithm~\ref{alg:error_construction}, and the decoding procedure in Algorithm~\ref{alg:error_decoding}. Both the construction and the decoding algorithm are in the same spirit as Algorithm~\ref{alg:construction} and Algorithm~\ref{alg:decoding} respectively, by learning $d = \left\lceil \alpha \log k \right \rceil$ bits at a time, instead of $1$. It is not hard to infer that no defective item will be left out of the $L$, the list output by Algorithm~\ref{alg:error_decoding}. What remains is to bound $|L|$. We will show that with probability $1 - e^{-C_1 k\log(n/k)}$, the list will always have $C_L k\log_k(n/k)$ items in it, where $C_1$ is an absolute constant.
 
\begin{lemma}

For any constant $C_1$ there exist choices of $C,C_L$ such that
\begin{multline*}
		\mathbb{P}\left[ \exists x \in \{0,1\}^n,|\mathrm{supp}(x)| \leq k\colon |\textsc{Identify-under-Errors}(M \odot x)| > C_L k \log_k(n/k) + k \right] \\< e^{-C_1 k \log (n/k)}, 
\end{multline*}
under the presence of $k \log k$ false positives.
 The randomness is over the functions $\{h_{\ell,r} \}_{(\ell,r) \in [H]\times[R]}$. Moreover, the running time of Algorithm~\ref{alg:error_decoding} is $D \cdot O(k \log (n/k)) = O(k^{1+\alpha} \log (n/k))$.
\end{lemma}

\begin{proof}

For a given $x \in \{0,1\}^{\log n}$, we define the trajectory $\tau_x$ of the items to be the set of all possible prefixes $p$ that might be inserted in $L$ at some point during the execution of the algorithm and will not get discarded in Line~\ref{lin:error_discard}. We also refer to the trajectory of the defective items, and denote it as $\tau_x'$, as the set of possible prefixes $p$ containing a defective item and not inserted in $L$ at some point during the execution of the algorithm. A prefix $p$ of length $\ell$ has $D=2^d$ children $\bigcup_{p' \in \{0,1\}^d} p \| p'$ of length $\ell+1$. We will imagine a degree-$D$ tree $\mathcal{T}$ rooted on the empty string. Moreover, the relation $T \subseteq_{\mathrm{tree}} \mathcal{T}$ will denote the fact that $T$ is a connected sub-tree of $\mathcal{T}$ rooted at the empty string. 

It obviously suffices to prove that $|\tau_x \setminus \tau_x'| \leq C_L k \log_k(n/k)$, for any $x$ and any choices of false positives. This ensures that the output of $\textsc{Identify-Under-Errors}(M \odot x)$ has of size $O(k \log_k(n/k))$.

We will prove a stronger fact. Consider sets $E^{(\ell,r)} \subseteq [Ck], (\ell,r) \in [H]\times[R]$, such 
\begin{align} \label{eqn:badcondition}
\forall \ell \in [H]\colon \sum_{ r \in [R]} |E^{(\ell,r)}| \leq k\log k.
\end{align}
Call such a collection of sets ``bad''.

The set $\bigcup_{ (\ell,r) \in [H]\times[R]} E^{(\ell,r)}$ will correspond to the set of false positives; i.e., the measurements that an adversary can corrupt in order to sabotage the decoding algorithm. Note that it suffices to prove the lemma with equality in \eqref{eqn:badcondition}, for all $\ell \in [H]$, since adding more false positives can only hurt the decoding algorithm. For a fixed $\ell$, counting $\bigcup_{r \in [R]} E^{(\ell,r)}$ corresponds to putting $k \log k$ balls in $R$ bins such that in each bin we can have at most $Ck$ balls. By relaxing the condition on the upper bound on the capacity of the bin, we get that the number of valid choices for the sets $E_{\ell,r}$ is upper-bounded by

	\[	{ k\log k + R-1 \choose R-1}^{H}.	\]

Let us fix now $\{E^{(\ell,r)}\}_{ (\ell,r) \in [H] \times [R]}$.
For $\ell \in [H], r\in[R]$ and a prefix $p$ of length $\log k + \ell \cdot d$, define Bernoulli random variable $Y_{p,r}$ such that

\[	Y_{p,r} = 1~\mathrm{iff}~ \left( (M^{(\ell,r)})_{h_{\ell,r}(p)} \cap \mathrm{supp}(x) \neq \emptyset \right)~\mathrm{or}~\left( h_{\ell,r}(p) \in \mathrm{E}^{(\ell,r)}\right).	\] 
In words, $Y_{p,r}$ captures an event where $p$ participates in a positive test in $M^{(\ell,r)}$ (i.e., either by the existence of a defective or by a false positive).

Moreover, let binary Bernoulli random variable $Y_p$ such that $Y_p = \mathrm{maj}_{r \in [R]} Y_{p,r}$. For a fixed $\ell$ there can be at most $R/4= C \log k/4$ repetitions $r\in [R]$ with $|E^{(\ell,r)}| > 4k$; call any $r$ which does not satisfy this ``desirable''. For $p$ not containing a defective item and a desirable $r$, we have that 
\[\mathbb{P}\left\{ Y_{p,r}=1 \right\} \leq \frac{5}{C}.\] An application of the Chernoff Bound for variables $\{Y_{p,r}\}_{r~\mathrm{desirable}}$ yields that $Y_p = 1$ with probability at most $(10/C)^{\log k}$ for sufficiently large $C $. 
Given the above considerations, we have that 

\begin{align*}
\mathbb{P} \left[ |\tau_x \setminus \tau_{x'} | > C_L k \log_k(n/k) \right] 
 &\leq \mathbb{P} \left[ \exists T \subseteq_{\mathrm{tree}} \mathcal{T}, |T| = (C_L +1)k\log_k(n/k) +k\colon \forall p \in T, Y_p =1 	\right] \\
&\leq \mathrm{Cat}_{(C_L+1)k\log_D(n/k) + k}^D \left( \frac{10}{C} \right)^{\log k \cdot C_L k \log_k(n/k) }\\
&\leq (ek^\alpha)^{(C_L+1)k \log_D(n/k) + k} \left( \frac{10}{C} \right)^{C_L k \log(n/k) } \\
&= 2^{ \alpha (C_L k \log_D (n/k) + k))\cdot \log (ek)} \left( \frac{10}{C} \right)^{C_L k \log(n/k)} \\
&\leq 2^{8\alpha C_Lk \log(n/k)}  \left( \frac{4}{C} \right)^{C_L k \log(n/k)}
\end{align*}
Now, 
\begin{align*}
&\mathbb{P} \left[ 
\left(\exists x \in \{0,1\}^n, |\mathrm{supp}(x)| \leq k,~\mathrm{and}~\mathrm{bad~collection~} \{E^{(\ell,r)}\} \right)\colon |\tau_x\setminus \tau_x'| > k \log_k(n/k) \right] \\
&\leq \underbrace{\left( \sum_{j=0}^k {n \choose k} \right)}_{\mathrm{choices~for~}x} \cdot \underbrace{{k \log k + R -1\choose R-1}^{ \log_D(n/k)}}_{\mathrm{choices~for~}E^{(\ell,r)}} \cdot 2^{8\alpha C_Lk \log(n/k)}  \left( \frac{10}{C} \right)^{C_L k \log(n/k)}\\
&\leq (k+1) \cdot {n \choose k} \cdot 2^{(k\log k + R-1) \cdot (\log_D (n/k)+1)} \cdot  2^{8\alpha C_Lk \log(n/k)}  \left( \frac{10}{C} \right)^{C_L k \log(n/k)} \\
&\leq (k+1) e^{k \log(en/k)}\cdot 2^{3k\log (n/k)} \cdot  2^{8\alpha C_Lk \log(n/k)}  \left( \frac{10}{C} \right)^{C_L k \log(n/k)/2} 
\end{align*}

It is not hard to see that for any $\alpha$ we can pick $C,C_L$ with $C < C_L$ such that the latter inequality is at most $e^{-Ck \log(n/k)}$. This ensures that in $L$ at most 
\[	D \cdot \left( \underbrace{O(k \log (n/k))}_{\mathrm{due~to~}\tau_x'} + \underbrace{O(k \log_k(n/k))}_{\mathrm{due~to~}\tau_x\setminus \tau_x'}\right) = O(k^{1+\alpha}\log(n/k))	\]
 prefixes will be inserted during the execution of the algorithm, and thus the bound on the running time follows. This finishes the proof. 
\end{proof}

Thus, the first claim of Theorem~\ref{thm:error_list} follows as discussed. We may obtain the second claim of Theorem~\ref{thm:error_list} by augmenting the matrix guaranteed by the first claim with the disjunct matrix guaranteed by Theorem~\ref{thm:error_tool} to filter out the non-defective items.

\begin{remark}
A list-disjunct matrix with $m=O(k\log (n/k) + \log(n/k) \cdot e_0)$ rows and $O(m \log m)$ decoding time follows easily by the reduction at the beginning of this subsection and Theorem~\ref{thm:main_tool}.
\end{remark}

\subsection{Proof of Theorem~\ref{thm:error_dis}}

We shall show how to set up a hashing scheme and then invoke the framework of~\cite{lnnt16}. That framework is similar in spirit with list-recoverable codes, albeit with linking between the lists, allowing faster decoding and optimal measurement complexity in that work, which focused on the heavy hitters problem in the for-each case. The idea of linking has also been used before in~\cite{glps17}. 
What we are essentially doing here is to construct a $(k,O(k),0,e_1)$-list-disjunct matrix with $O(k^2 \log n + ke_1)$ rows. This is of course unacceptable for list-disjunct matrices but fine for disjunct matrices. Our approach is closely related to the implicit reduction of~\cite{indyk2010efficiently} to list-disjunct matrices, albeit with the techniques of~\cite{lnnt16} and an added twist (two-layer hashing) to make the decoding time nearly linear (instead of polynomial) in $m$.  The necessity for two-layer hashing will add additional complexity to the whole argument.

Formally, the technical contribution of this subsection is the following Lemma.
\begin{lemma}\label{lem:error_linking}
There exists a $(k,O(k),e_0,0)$-error-correcting list-disjunct matrix with $m = O(k^2 \log n + ke_1)$ rows, which allows decoding in time $m \cdot \poly(\log n)$.
\end{lemma}

Indeed, using Lemma~\ref{lem:error_linking} the non-defective items can be filtered out in the desired time bound by using Theorem~\ref{thm:error_tool}. We thus focus on proving Lemma~\ref{lem:error_linking}. We re-iterate that the above result, although quite undesirable for list-disjunct matrices, suffices for our application to disjunct matrices.

\begin{proof}[Proof of Lemma~\ref{lem:error_linking}]
First of all, by a similar reduction as in Theorem~\ref{thm:error_list} we can assume that $e_1 \leq k\log n$. Let $C$ be a large enough absolute constant and let $R = Ck$ and $Q = \log n/ \log \log n$.

Let $\mathrm{enc}\colon \{0, 1\}^{\log n}\rightarrow  \{0, 1\}^{O(\log n)}$ be the encoding function of a constant-rate error-correcting code that corrects a constant fraction of errors in linear time; such a code is available in~\cite{s96}. We split $\mathrm{enc}(i)$ into $Q$ blocks of length $C \log \log n$, and denote the $q$-th block by $\mathrm{enc}(i)_q$. Let us pick also a $d$-regular connected expander, which we shall call $F$, on the vertex set $[Q]$ for some constant $d$. For $q \in [Q]$ we let $\Gamma(q) \subseteq [Q]$ be the set of neighbors of $q$. 

We shall pick the following two collections of hash functions.
\begin{enumerate}
\item $g_\rho\colon \{0,1\}^{O(\log n)} \rightarrow [Ck/\log n], \forall \rho \in [R]$.
\item $h_{\rho,q}\colon \{0,1\}^{O(\log n)} \rightarrow \{0,1\}^{C \log \log n}, \forall (\rho,q) \in [R] \times [Q]$.
\end{enumerate}

For a set $S \subseteq \{0,1\}^{O(\log n)}$, we denote $h_{\rho,q}(S) = \bigcup_{i \in S} h_{\rho,q}(i)$. Moreover, \[g^{-1}_\rho(b) = \{ i \in \{0,1\}^{O(\log n)}\colon g_\rho(i) = b\}\] for $b \in [Ck/\log n]$.

The following lemma constitutes the crux of the reduction.
\begin{claim}\label{claim:error_dis_lemma}
There exists a choice of functions $\{g_\rho, h_{\rho,q}\}_{ (\rho,q) \in [R] \times [Q]}$ such that the following holds.
For all $S \subseteq \{0,1\}^{O(\log n)}, |S| \leq k$ and $j \in S$ there exist at least $3R/4$ repetitions $\rho \in [R]$ for which there exist at least $9Q/10$ $q\in[Q]$ such that
\begin{enumerate}
\item $|T| \leq \log n$,
 \item $h_{\rho,q}(j) \notin h_{\rho,q}(T) $,
\end{enumerate}

where $b = g_\rho(j)$ and $T = g_\rho^{-1}(b) \cap \left( S \setminus \{j\}\right)$.
\end{claim}

\begin{proof}
We shall pick every function at random. Fix $S,j$ as in the condition of the lemma.
For $\rho \in [R], b \in [Ck/\log n]$  note that $|g^{-1}_\rho(b) \cap S|$ is the sum of $k$ Bernoulli random variables with expectation $|S| \cdot (\log n/( C k)) \leq \log n / C$. By an application of the Chernoff bound we get $|g^{-1}_\rho(g_\rho(j)) \cap (S\setminus\{j\})| \leq \log n$ with probability $1 - e^{-\Theta(\log n)}$. The probability that in more than $R/4 -1$ repetitions, we have that $|g^{-1}_\rho(g_\rho(j)) \cap (S\setminus \{j\})| > \log n$ is \[e^{-\Theta(R \log n)} = e^{-\Theta(k \log n)}\] by a Chernoff bound; the constant in the exponent can be tuned arbitrarily large by setting $C$ sufficiently large. A union bound over all pairs $(j,S)$ with \[S \subseteq \{0,1\}^{\log n}, |S| \leq k, j \in S,\] the number of which is at most \[k \cdot \sum_{i=0} {n \choose k} \leq k^2 {n \choose k},\] yields item~1.

Let us condition on item~1 being true. Let us choose one of the at least $3R/4 + 1$ repetitions $\rho \in [R]$ for which bullet item~1 holds for pair $(j,S)$. For that repetition, we have that item~2 fails with probability $e^{-(C-1)\log \log n}$. Hence with probability \[e^{-(C-1)\log \log n \cdot \Omega(Q)} = e^{-\Theta(\log n)},\] item~2 holds for at least $9/10$ of the indices $q \in [Q]$. Thus, by the Chernoff bound over all $3R/4+1$ aforementioned repetitions we get that with probability $e^{-\Omega(k \log n)}$ there exist at least $R/2$ repetitions such that items 1~and~2 simultaneously hold for pair $(j,S)$. We may now take a union-bound over all $(j,S)$ such that $S \subseteq \{0,1\}^{\log n}, |S| \leq k, j \in S$ to conclude the claim.

\end{proof}

For every $\rho \in [R], b \in [Ck/\log n]$, and $q \in [Q]$ let us now partition $\{0,1\}^{O(\log n)} \cap g^{-1}_\rho(b)$ according to 
\[\mathcal{O}_{\rho,b,q}(i) = h_{\rho,q}(i) \| \mathrm{enc}(i)_q \| h_{\rho,\Gamma(q)_1}(i) \| \cdots \| h_{\rho,\Gamma(q)_d}(i),\]
i.e., every $i \in \{0,1\}^{O(\log n)} \cap g^{-1}_\rho(b)$ belongs to partition $\mathcal{O}_{\rho,b,q}(i)$. Note that for each $\rho,b$ there are $2^{(d+2)C \log \log n}$ such partitions. 

For every $(\rho,b,q) \in [R] \times [Ck/\log n ]\times [Q]$ we now pick a $(\log n, 2^{(d+2)C \log \log n},0, 20\log \log n)$ error-correcting list-disjunct matrix $M^{(\rho,b,q)}$ guaranteed by Theorem~\ref{thm:error_tool} on universe size $2^{(d+2)C \log \log n}$ and extend it over $\{0,1\}^{O(\log n)}$ by putting every $i \in g^{-1}_\rho(b)$ to the measurement corresponding to $\mathcal{O}_{\rho,b,q}(i)$. In other words, for each $(\rho,b,q)$ we group together coordinates with the same $Q_{\rho,b,q}(i)$ value, and pick a list-disjunct matrix over the smaller universe of size $2^{(d+2)C \log \log n}$. The total number of rows is 
\[	\underbrace{ k  \cdot  (C k/\log n )  \cdot Q}_{\mathrm{choices~for~}(\rho,b,q)} \cdot \underbrace{ O( \log n \cdot \log \log n )}_{\mathrm{number~of~rows~of~}M^{(\rho,b,q)}} = O(k^2 \log n).	\]

\begin{claim}\label{claim:false_negatives}
For any adversarial choice of $k \log n$ false negatives and for every $j \in \{0,1\}^{\log n}$, there exist at least $\frac{3R}{4}+1$ indices $\rho \in [R]$ for which the following holds. Denoting $ b = g_\rho(\mathrm{enc}(j))$, at least $3Q/4$ of the list-disjunct matrices $M^{(\rho,b,q)}$ with $q\in [Q]$ receive less than $20 \log \log n$ adversarial errors. 
\end{claim}

\begin{proof}
Assume that this was not the case. Then the number of false negatives would be strictly more than	\[	\frac{R}{4} \cdot \frac{Q}{4} \cdot (20 \log \log n) > k \log n,	\]
which is a contradiction.
\end{proof}

We are now in a position to discuss the decoder. For every $(\rho,b) \in [R] \times [Ck/ \log n]$, we shall first show how to find a list $L_{\rho,b}$ of size $O(k)$. Then we shall keep the elements $j $ that appear in at least $R/2$ of the lists $L_{\rho,b}$. For each such $j$ we may return $i = \mathrm{enc}^{-1}(j)$, which can be done in $O( \log n)$ time by the guarantee of~\cite{s96}. Since for fixed $R$ and across $b$ every $j \in \{0,1\}^{O(\log n)}$ is associated with exactly one choice of $b$, an averaging argument shows that the number of returned $j$ can be $O(k)$. What needs to be shown is that we do not miss any $j = \mathrm{enc}(i)$ where $i$ is defective.

\paragraph{Obtaining $L_{\rho,b}$.} We run the naive decoder of Definition~\ref{def:naive} on every list disjunct matrix $M^{(\rho,b,q)}$ to get a list $L_{\rho,b,q}$ of size at most $2k$. Our approach closely follows~\cite{lnnt16} from now on. Fix $\rho,b$. We build a graph $G$ on vertex set $[Q] \times [2^{(d+2)C\log \log n}]$. For every $q \in [Q]$ we split every $s \in L_{\rho,b,q}$ (which is a binary string of length $(d+2)C\log\log n$) to $d+2$ sub-strings of length $C \log \log n$, henceforth denoted $s_1,s_2,\ldots,s_{d+2}$. We say that $s_2$ is the ``name'' of $s$. If there are multiple elements $s$ with the same name we just keep one of those. We say that $s_2$ ``suggests'' an edge $e$ connecting $(q,s_2)$ to $(q,s_\ell), \forall \ell \in\{1,3,4,\ldots,d\}$. If $s_\ell$ is the name of some other element in $L^{(\rho,b,q)}$ and suggests also $e$, we add $e$ to $G$. This will result in a graph with at most \[(d/2) \cdot Q \cdot (2\log n) = O(\log^2 n/ \log \log n)\] edges. We restrict $G$ to the union of non-isolated vertices, ensuring that it has $O(\log^2 n / \log \log n)$ vertices. As in~\cite{lnnt16} we may now perform spectral clustering (\cite[Theorem 1]{lnnt16}) in $G$ to find all $\epsilon_0$-spectral clusters for some $\epsilon_0$ constant.
As in~\cite{lnnt16} this will yield a set of corrupted codewords in $\{0,1\}^{O(\log n)}$, from which we shall decode to obtain $L_{\rho,b}$.

To prove correctness, fix a set of defective items $S$ and $i \in S$. Using Claim~\ref{claim:error_dis_lemma} and Claim~\ref{claim:false_negatives}, we may infer that there exist at least $R/2$ repetitions $\rho \in [R]$ such that $j = \mathrm{enc}(i)\in L_{\rho,g_\rho(j)}$; the proof is totally analogous to the proof of~\cite[Theorem~2]{lnnt16}. This happens in particular because there exist at least $R/2$ indices $\rho \in [R]$ for which there exist at least $3Q/4$ elements $q \in [Q]$ such that i) the list-disjunct matrix $M^{(\rho,g_\rho(j),q)}$ will be decoded correctly for it shall receive less than $20\log \log n$ false negatives (Claim~\ref{claim:false_negatives}), and ii) $j = \mathrm{enc}(i)$ will not ``collide'' with another defective item (Claim~\ref{claim:error_dis_lemma}) in that matrix. The above two conditions constitute the analog of~\cite[Lemma~1]{lnnt16}. This ensures that $\mathrm{enc}(i)$ will be present in $L$ for every $i \in S$. Hence, no defective item will be lost, and since $|L|= O(k)$ as argued, this finishes the proof of the lemma.
\end{proof}
Putting together Lemma~\ref{lem:error_linking} with Theorem~\ref{thm:error_tool} we obtain the Theorem~\ref{thm:error_dis}.
We re-iterate that the above approach \emph{does not} perform well with respect to false positives, see also discussion at the end of this Section. It is also quite specific to the construction of disjunct matrices since it yields a quadratic bound in the number of measurement (roughly speaking, this happens because one needs to repeat $\Theta(k)$ times in order to boost the success probability so that a union-bound over all possible sets of defectives is possible). 

\paragraph{Does there exist a simpler approach for decoding disjunct matrices?}  

It is natural to wonder whether there exists a simpler way of proving Theorem~\ref{thm:error_dis}, avoiding the spectral graph theory framework of \cite{lnnt16} and its ad-hoc incorporation via the shown complicated two-layer hashing scheme. We present the following stand-alone reduction from disjunct matrices that can correct up to $e_1$ false negatives  to $(k,n,0,e_1)$-error-correcting list-disjunct matrices, which could be useful in such an attempt. Implicit and explicit reductions from disjunct to list-disjunct matrices appear also in \cite{ch2009noise,indyk2010efficiently}, but in order to work out, they demand list-recovery technology to combine the answers from the list-disjunct matrices. However, the reduction we present here does not demand any such combination.

\begin{lemma}\label{lem:reduction}
Assume there exists a $(k,O(k),0,e_1)$-error-correcting list-disjunct matrix with $R(k,n,e_1)$ rows which is decodable in time $ \poly(R(k,n,e_1))$, in the regime $k \leq c \log n$. Then there exists a $k$-disjunct matrix that can tolerate up to $e_1$ false negatives, with \[m = O((k^2/ \log n) \cdot R(\log n,n,e_1/k))\] for $k > c\log n$ rows. Furthermore, the matrix is decodable in time \[O((k^2/ \log n) \cdot \poly(R(\log n,n,e_1/k))).\]-
In particular, if $R(k,n,e_1) = O(k \log (n/k) + k e_1)$, we obtain an error-correcting $(k,n,e_1)$-disjunct matrix with $m = O(k^2 \log n + ke_1)$ rows and which is efficiently decodable in $m \cdot \poly(\log n)$ time (for all values of $k$).
\end{lemma}

\begin{proof}
Let $k > c \log n$. Let $C$ be a sufficiently large constant. For $\rho \in [C k]$, consider a random hash function $h_\rho\colon [n] \rightarrow [ 10k / \log n]$. For each $\rho \in [Ck], b \in [10 k/\log n]$ we consider a $(\log n,|h_\rho^{-1}(b)|,0,e_1 / k )$-error-correcting list-disjunct matrix over universe $h^{-1}_\rho(b)$. Let that matrix be $M^{(\rho, b)}$. The total number of rows over all $M^{(\rho,b)}$ is 
\[	Ck \cdot (10k / \log n) \cdot R(\log n,n,e_1/k).\]

Fix a set $S \subseteq [n]$ of defective items. The probability that $i \in S$ in repetition $\rho$ satisfies $|h^{-1}_\rho(h_\rho(i)) \cap S| > \log n$ is at most $e{-2\log }$ by the Chernoff bound (as in the previous proof, $|h_\rho^{-1}(h_\rho(i))\cap S|$ is controlled by the sum of $k-1$ independent Bernoulli random variables of expectation $\frac{\log n}{10k}$). The probability that in \emph{most repetitions} $\rho$ we have $|h^{-1}_\rho(h_\rho(i)) \cap S| > \log n$ is at most $e^{-C k  \log n / 10}$. By a union-bound over all $ k \cdot \sum_{j=0}^k {n \choose k} \leq k^2 {n \choose k}$ pairs $(i,S)$ with $|S| \leq k, i \in S$, we get that for every such pair $(i,S)$  it holds that $|h_\rho^{-1}(h_\rho(i))\cap S| \leq \log n$ in most repetitions $\rho$. Let us call this property ``expansion''. For every $i \in S$ there can be at most $k $ repetitions $\rho$ for which $M^{(\rho, h_\rho(i))}$ receives more than $e_1/k$ errors. This means that we can perform the following in order to find a list $L$ that contains all defective items: Run the decoding procedure associated with every $M^{(\rho,b)}$, halting after $\poly(R(\log n, n ,e_1/k))$ steps, returning all elements that appeared at list $Ck/2$ times. By the expansion property, for every $i \in S$ and every $\rho \in [Ck]$, the decoding procedure of $M^{(\rho,h_\rho(i))}$ will run on a valid instance of sparsity less than $\log n$ and with at most $e_1/k$ errors. Hence, it will run correctly, yielding the desired result. The total running time is at most

\[	(C k) \cdot (10k / \log n) \cdot \poly(R( \log n , n, e_1/k)).	\]
We shall bound the size of $L$ by $O(k)$. Note that the number of pairs $(i,\rho)$, where $i$ is not defective but $i$ was returned by the decoding procedure on $M^{(\rho,h_\rho(i))}$ is $ O(\log n ) \cdot (10k / \log n) \cdot Ck = O(k^2)$. Thus, a simple averaging argument shows that there can be at most $O(k)$ coordinates $i$ which are non-defectives but were returned by $M^{(\rho,h_\rho(i))}$ in the majority of repetitions $\rho$. From that, we conclude that $L = O(k)$. Using the error-correcting disjunct matrix guaranteed by Theorem~\ref{thm:error_tool} (part~2), we may find exactly all the defective items using point-queries.

When $R(n,k,e_1) = O(k \log (n/k) + ke_1)$, the above reduction in the regime	 $k > c \log n$ gives the desired result. In the regime $k \leq c \log n$ we may use this list-disjunct matrix and the reduction from $e_1$ false negatives to $k \log n$ false negatives, to find a list of size $2k$ that can be filtered out using the matrix in Theorem~\ref{thm:error_tool}. The main observation for establishing the running time is again that $\poly(k,\log n) = k^2 \poly(\log n)$, which gives the desired result.
\end{proof}

\definecolor{myblue}{RGB}{80,80,160}
\definecolor{mygreen}{RGB}{80,160,80}

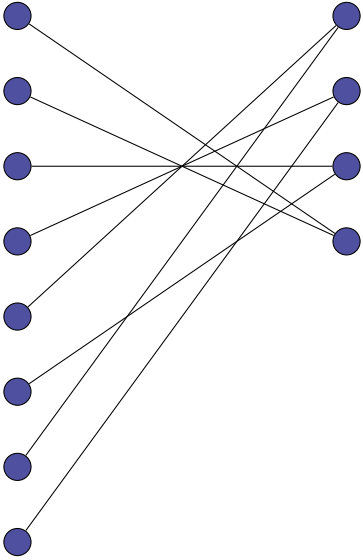
\begin{figure}
\begin{center}
\begin{tikzpicture}
   \graph[nodes={draw, circle,fill=myblue}, radius=.5cm,
           empty nodes, branch down=1 cm,
           grow right sep=4cm] {subgraph I_nm [V={a, b, c, d, e,f,g,h}, W={1,...,4}];
  a -- { 4};
  b -- { 4 };
  c -- { 3 };
  d -- { 2 };
  e -- { 1};
  f -- { 3};
  g -- { 1};
  h -- { 2};
};
\end{tikzpicture}
\end{center}
\caption{A schematic illustration of a \emph{single} repetition $\rho$ of the hashing scheme in Lemma~\ref{lem:reduction}. The left layer corresponds to the universe $[n]$, and the middle layer to buckets $b \in [Ck/\log n]$. Every $i \in [n]$ is hashed to one out of $Ck/ \log n$ possible buckets, using $h_\rho\colon [n] \rightarrow [Ck/\log n]$. For each bucket $b$, a list-disjunct matrix with logarithmic sparsity and error tolerance $e_1/k$ is constructed over universe $h_\rho^{-1}(b)$.}
\end{figure}

Note that the above reduction, as well the hashing scheme in Theorem~\ref{thm:error_dis}, fail in presence of false positives. An adversary which has $e_0 = \Theta(k^2 \log n)$ false positives in their possession can choose a defective item $i$ and put $\Omega(k \log n)$ false positives in each matrix $M^{(\rho,h_\rho(i))}$. Note that each such matrix has $O(\log^2n)$ rows, and hence for the interesting case of $k \gg \log n$ the adversary can mess up all the list-disjunct matrices where $i$ participates in. It is unclear how to decode from such a scenario, if at all possible. A natural approach would be to discard each such $M^{(\rho,h_\rho(i))}$, because it would appear as a matrix overloaded with false positives. However, this would lead to not detecting $i$. 

\section{$\ell_2/\ell_2$ Compressed Sensing}\label{sec:ell2}

This section is devoted to proving Theorem~\ref{thm:ell2}.
In what follows, for a vector $x \in \mathbb{R}^n$ and a set $S$, we let $x_S$ be the vector that is obtained by zeroing out every $i \notin S$. We let $x_{-k}$ to be the vector that is obtained after zeroing out the largest $k$ in magnitude coordinates in $x$, breaking ties arbitrarily. 
We shall use the anti-concentration of Gaussians and in particular the following fact.

\begin{claim}\label{claim:anti}
There exists an absolute constant $C_g>1$ such that
\[	\mathbb{P}_{g \sim \mathcal{N}(0,1)}\left\{ \frac{1}{G_g} \leq |g| \leq C_g \right\} \geq 9/10.	\]
\end{claim}

\begin{table}[!ht]
\centering
\begin{tabular}{|l|l|l|l|}
\hline
  Notation &  Statement \\ \hline
  $C_g$ & Governs the concentration and anti-concentration of Gaussians\\ \hline
  $C$ &  Governs the number of rows per matrix $M^{(\ell,r)}$\\ \hline
  $C_T$ & Governs the number of non-discarded elements in Line~\ref{lin:ell2_discard}\\ \hline
  $B$ &  Governs the size of sets $E^{(\ell,r)}$ in Claim~\ref{claim:2}\\ \hline
  $C_F$ &  Governs the size of the forest $|\mathcal{F}|$ in Claim~\ref{claim:2} \\ \hline
\end{tabular}\caption{Summary of constants in Section~\ref{sec:ell2}. We shall set $C_g \ll B \ll C \ll C_F$.}
\end{table}

Anti-concentration of Gaussians for $\ell_2/\ell_2$ compressed has been also used in \cite{NakosS19}, but for a crucially different reason. That work also demands  a much stronger anti-concentration property of the Gaussian distribution (the point there was to avoid arguing via weak identification systems in the first place in order to achieve optimality in decoding time and column sparsity, something which is not our focus here). Though they might seem relevant at a first glance, our approach and the one in  \cite{NakosS19} are quite different on many levels, the most obvious being the fact that \cite{NakosS19} \emph{does not} build a weak identification system and does not perform the standard iterative loop employed in the compressed sensing literature. Note also that Claim~\ref{claim:anti} does not really use anything but the fact that Gaussian is a continuous distribution.

It is possible to use the machinery of FT-mollification~\cite{KNW10} to work with discretized versions of Gaussians which can be stored in small space. However, in order to keep the exposition as elementary as possible, here we do not elaborate on how to do so.

\begin{algorithm}[!ht]
\caption{Construction of matrix for Lemma~\ref{lem:ell2_identify}}
\label{alg:ell2_construction}
\begin{algorithmic}[1]
\Procedure{\textsc{CreateMatrix}}{$k$}
	\State $R \leftarrow C\left(\frac{\log(n/k)}{ \log \log(n/k)} + \frac{\log(1/\delta)}{ k}\right)$
	\State $D \leftarrow \frac{\log(n/k)}{ \log \log (n/k)}$  \Comment{Round to the closest power of $2$.}
	\State $d \leftarrow \log D$
	\State $H \leftarrow \log_D (n/k)$	\Comment{ $H = \Theta\left(\frac{\log(n/k)}{ \log \log (n/k)}\right)$}
	\For { $\ell=0$ to $H$}
		\For { $r=0$ to $R-1$}
		\State Pick hash function $h_{\ell,r}\colon \{0,1\}^{\log k + \ell\cdot  d} \rightarrow [Ck/\epsilon]$.
			\For { $i \in \{0,1\}^{\log n}$ }
				\State $ q \leftarrow h_{\ell,r}(\bpref_{\log k + \ell \cdot d }(i))$
				\State Pick $g\sim\mathcal{N}(0,1)$ with fresh randomness.
				\State $M^{(\ell,r)}_{ q, i } = g$		
			\EndFor
		\EndFor	\Comment{Every non-initialized entry of $M^{(\ell,r)}$ is $0$.}
	\EndFor		
	\State \Return $M$ as the vertical concatenation of $M^{(\ell,r)}, (\ell, r) \in [H]\times [R]$.
\EndProcedure
\end{algorithmic}
\end{algorithm}

\begin{algorithm}[!ht]
\caption{Decoding procedure associated with Lemma~\ref{lem:ell2_identify}}
\label{alg:ell2_decoding}
\begin{algorithmic}[1]\Procedure{\textsc{Decode-CompressedSensing}}{$y$} 
	\State $R \leftarrow C\left(\frac{\log(n/k)}{ \log \log(n/k)} + \frac{\log(1/\delta)}{k}\right)$
	\State $D \leftarrow  \frac{\log(n/k)}{ \log \log (n/k)}$  \Comment{Round to the closest power of $2$.}
	\State $d \leftarrow \log D$
	\State $H \leftarrow \log_D (n/k) + 1$	\Comment{ $H = \Theta\left(\frac{\log(n/k) }{\log \log (n/k)}\right)$}
	\State $L \leftarrow \{0,1\}^{\log k} $
	\For { $\ell=0$ to $H$}
		\For {$i \in L$}
			\State $L_p \leftarrow \emptyset$

			\For { $r=0$ to $R-1$} 
				\State $q \leftarrow  h_{\ell,r}(p)$ \Comment{Find in which row of $M^{(\ell,r)}$ element $i$ is set to $1$.} 
				\State $L_p \leftarrow L_p \bigcup M_q^{(\ell,r)}  x$  \Comment{Fetch the  corresponding entry.}	\label{lin:gaussian}
			\EndFor
			\State $\mathrm{est}_{p} \leftarrow \mathrm{median}_{ z \in L_p} |z|$
		\EndFor
		\State Discard all $p \in L$ apart from those with the largest $C_T (k/\epsilon)\log(n/k)~\mathrm{est}_{p}$ values.	\label{lin:ell2_discard}
		\If {$\ell = H-1$}
			\State \Return $L$.
		\EndIf
		\For {$i \in L$}
			\State Add $ \bigcup_{p' \in \{0,1\}^d} p \| p' $ to $L$	.\Comment{Expand $L$.}
			\State Discard $p$ from $L$.
		\EndFor
	\EndFor
\EndProcedure
\end{algorithmic}
\end{algorithm}

\begin{lemma}\label{lem:ell2_identify}
There exists a randomized construction of a matrix $M \in \mathbb{R}^{m \times n}$ with 
\[m = O( (k/\epsilon) \log (n/k) + \frac{1}{\epsilon} \cdot \frac{\log(n/k)}{ \log \log (n/k)} \log(1/\delta)),\] such that given $y= M x$ we can find in time $O(m \log^2 n)$ a list (set) $L\subseteq [n]$ of size $C_T (k/\epsilon)\log(n/k)$ which with probability $1-\frac{\delta}{2}$ satisfies the following. There exists $T\subseteq L$  with $|T| \leq k$, such that 
\[\|(x-x_T)_{-k/10}\|_2^2 \leq \big(1+ \frac{9\epsilon}{10}\big) \|x_{-k}\|_2^2.	\]
\end{lemma}

\begin{proof}
The construction of $M$ appears in Algorithm~\ref{alg:ell2_construction}. The number of rows is 

\[	H \cdot R \cdot (Ck/\epsilon) = O((k/\epsilon) \log(n/k) + \frac{1}{\epsilon} \cdot \frac{\log(n/k)}{ \log \log (n/k)} \log(1/\delta)).	\]
The decoding procedure on $M$ is depicted in Algorithm~\ref{alg:ell2_decoding}. The running time is 
\[	H \cdot ( D \cdot C_T (k/\epsilon)\log(n/k) )\cdot R \cdot O(1) = O(m\log^2 n).	\]

We now prove correctness. Let \[\mathcal{H} = \{ i \in \{0,1\}^{\log n}\colon |x_i|^2 \geq (\epsilon/k) \|x_{-k}\|_2~\mathrm{and}~(x_{-k})_i = 0\},\] i.e., those coordinates among the top $k$ whose magnitudes are at least $\sqrt{\epsilon/k}\|x_{-k}\|_2^2$.
For the sake of the analysis (and for avoiding making a lengthy induction argument which would be needed otherwise), we will define $\mathrm{est}_{p}$ for every length-$\ell$ prefix $p \in \{0,1\}^{\log k + \ell \cdot d}$, exactly the same way as in Algorithm~\ref{alg:ell2_construction}. Note that the algorithm avoids computing all $\mathrm{est}_{p}$ by considering a very small subset of all prefixes $p$, namely only $O((k/\epsilon)\log^3(n/k))$ prefixes.

\begin{claim}\label{claim:1}
With probability $1-\frac{\delta}{6}$, the following holds for all but $(k/10)$ elements $i \in \mathcal{H}$. For prefix $p_\ell= \bpref_{\log k + \ell \cdot d}(i)$ we have that $\mathrm{est}_{p_\ell} \geq (1/G_g) \sqrt{\epsilon/k} \|x_{-k}\|_2$ for all $\ell \in [H]$.
\end{claim}
\begin{proof}
Fix $i \in \mathcal{H}$ and consider $p_\ell$ such that $ \bpref_{\log + \ell \cdot d}(i) = p$. For $r \in [R]$, the random variable $M_{h_{\ell,r}(p)}^{(\ell,r)}  x$ is distributed as a Gaussian with variance at least $(\epsilon/k) \|x_{-k}\|_2^2$. Claim~\ref{claim:anti} implies that $|M_{h_{\ell,r}(p)}^{(\ell,r)} x|$ will be at least $(1/G_g) \cdot \sqrt{\epsilon/k} \|x_{-k}\|_2$ with probability $9/10$, and hence a Chernoff bound across all $r \in [R]$ gives that with probability \[\eta := 1 - e^{-c\log(1/\delta)/k -c\log (n/k)/\log \log(n/k)}\] it holds 
that $\mathrm{est}_{p_\ell}$ will be at least $(1/G_g) \cdot \sqrt{\epsilon/k} \|x_{-k}\|_2$, for some absolute constant $c$. Now, a union bound over all $\ell \in [H]$ shows that with probability $\eta$, 
it will be the case that \[\mathrm{est}_{p_\ell} \geq (1/G_g) \cdot \sqrt{\epsilon/k} \|x_{-k}\|_2\] for all $\ell \in [H]$ and $p = \bpref_{\log k + \ell \cdot d}(i)$. Consider now Bernoulli random variables $Y_i$ for $i \in \mathcal{H}$ such that

	\[	Y_i = 1~\mathrm{iff}~ \exists \ell \in [H], p_\ell = \bpref_{\log k + \ell \cdot d}(i)\colon |\mathrm{est}_{p_\ell}| < \frac{1}{C_g} \sqrt{\epsilon/k} \|x_{-k}\|_2.		\]

The above discussion gives that $Y_i = 1$ with probability $\eta$. 
We would like to apply the additive form of the Chernoff bound in order to argue that with probability $1-\frac{\delta}{6}$, at most $(k/10)$ of the $Y_i$ are $1$, from which the claim would then follow. What prevents us from doing so is the fact that the $Y_i$ are not independent. To circumvent that, we shall make a coupling argument and stochastically dominate the $\{Y_i\}_{i \in \mathcal{H}}$ by another set of Bernoulli random variables $\{Z_i\}_{i \in \mathcal{H}}$ which have larger expectations and are independent. As an intuitive explanation for that, note that if two $i,i' \in \mathcal{H}$ satisfy $p_\ell=\bpref_{\log k + \ell \cdot d}(i) = \bpref_{\log k + \ell \cdot d}(i')$, then the probability that $\mathrm{est}_{p_\ell}$ is smaller than the threshold only increases, since it is a median of the absolute values of $R$ Gaussians with twice the variance. A similar conclusion holds in the case where they have different prefixes and their prefixes participate in the same measurement. Thus, we may first define variables $W_{i,\ell}$ for $ i \in \mathcal{H}, \ell \in [H]$ as 
	\[	W_{i,\ell} = \mathrm{median}_{r \in [R]}|\mathcal{N}(0,|x_i|^2)|,	\]
where the Gaussians are independent across the $R$ repetitions. As a next step let us define 
		\[	Z_i = 1 ~\mathrm{iff}~\exists \ell \in [H]\colon W_{i,\ell} < \frac{1}{C_g} \sqrt{\epsilon/k}\|x_{-k}\|_2.	\]

Note that $Z_i$ are jointly independent, and they stochastically dominate $Y_i$; i.e., $\mathbb{P}\left\{ Y_i = 1\right\} \leq \mathbb{P}\left\{Z_i = 1 \right\}.$ Thus, the same analysis as above holds for $Z_i$, and we can now safely apply the Chernoff bound on $Z_i$ to conclude the claim.
\end{proof}
We also need the following claim.
\begin{claim}\label{claim:2}
A prefix $p$ of length $\log k + \ell \cdot d$ is called a ``potentially false prefix'' if 
\begin{enumerate}
\item (light prefix) For $\mathcal{I} := \{i \in \{0,1\}^{\log n}\colon \bpref_{\log k + \ell \cdot d}(i)=p \}$, $\|x_{\mathcal{I}}\|_2 < \frac{1}{\sqrt{2}C_g}\sqrt{\frac{\epsilon}{k}} \|x_{-k}\|_2$, and
\item (large estimator) $\mathrm{est}_{p} > \frac{1}{C_g} \sqrt{\frac{\epsilon}{k}}\|x_{-k}\|_2$.
\end{enumerate}

With probability $1-\frac{\delta}{2}$, the conclusion of Claim~\ref{claim:1} holds and there are at most $O((k/\epsilon) \log (n/k))$ potentially false prefixes inserted in $L$ during the execution of the algorithm which are not discarded in Line~\ref{lin:ell2_discard}. 
\end{claim}
\begin{proof}

In order to analyze $\mathrm{est}_{p}$ for potentially false prefixes, let us define \[\tilde{\mathrm{est}_p} = \mathrm{median}_{r\in[R]} \|x_{h_{\ell,r}^{-1}(h_{\ell,r}(p))}\|_2^2,\] where $\ell$ is the length of $p$. Note that the random variable $\tilde{\mathrm{est}_{p}}$ is exactly the \emph{variance} of the random variable $\mathrm{est}_p$; i.e., $\mathrm{est}_p \sim \mathcal{N}(0,\tilde{\mathrm{est}}_p)$. 

We may proceed along the same lines as in the arguments in Theorem~\ref{thm:main_tool} and Theorem~\ref{thm:error_list}, albeit with some careful, yet non-trivial twists. First of all, we set $D = \Theta(\log(n/k)/\log\log(n/k))$ and $R = \Theta(\log \log(n/k))$ whereas the scheme in Theorem~\ref{thm:main_tool} we chose $D = 2, R = 1$ and in Theorem~\ref{thm:error_list} we chose $D = \Theta(k^{\alpha}), R = \Theta(\log k)$.

Consider sets (to be specified later) $E^{(\ell,r)} \subseteq [Ck/\epsilon]$ with $|E^{(\ell,r)}| \leq B k/\epsilon$, for all $(\ell,r) \in [H] \times [R]$. Here $B$ is an absolute constant to be chosen later. A prefix of length $\ell$ is called ``undesirable'' if $h_{\ell,r}(p) \in E^{(\ell,r)}$ for more than $R/10$ repetitions $r \in [R]$. Walking through the proof of Theorem~\ref{thm:error_list}, we can infer that the probability that there exist choices of the sets $E^{(\ell,r)}$ and a $(C_F (k/\epsilon)\log(n/k)))$-sized forest $\mathcal{F}$ of $k$ trees rooted at all binary strings of length $k$, such that at least $|\mathcal{F}|/10$ prefixes $p \in \mathcal{F}$ are undesirable is at most
\begin{align*}	
\underbrace{{Ck/\epsilon \choose B k/\epsilon}^{H\cdot R}}_{ \mathrm{choices~for~}E^{(\ell,r)} }
 \cdot \underbrace{\mathrm{Cat}^D_{|\mathcal{F}| + k\log(n/k) + k}}_{\mathrm{choices~for~}\mathcal{F}}\cdot \underbrace{ {|\mathcal{F}| \choose \frac{1}{10}|\mathcal{F}|}}_{\mathrm{choices~of~undesirable~prefixes}} 
\cdot \underbrace{\left( {R \choose \frac{R}{10}} \left(\frac{B}{C}\right)^{\frac{R}{10}}\right)^{\frac{1}{10}|\mathcal{F}|}}_{probability~of~undesirability~of~prefixes}.
\end{align*}

Similarly to the proof of Theorem~\ref{thm:main_tool} and the proof of Theorem~\ref{thm:error_list}, using Lemma~\ref{lem:estimation_catalan} and the fact that ${a \choose b} \leq (ae/b)^b$  we can set the constants so that the latter probability is less than $\frac{\delta}{6}$. In particular, we may set $C/B$ to be an arbitrarily large constant, and $C_F$ sufficiently large with respect to $C$. From now on, let us condition on this event.

Given the above let us define for $(\ell,r) \in [H]\times[R]$ the set $E^{(\ell,r)}$ as

	\[		E^{(\ell,r)} = \{ b\in [Ck/\epsilon]\colon \|x_{h_{\ell,r}^{-1}(b)}\|_2 \geq \frac{1}{\sqrt{2}C_g}\sqrt{\epsilon/k} \|x_{-k}\|_2\}.	\]

Clearly, $|E^{(\ell,r)}| \leq k + 2C_g^2 \cdot (k/\epsilon)$, so at this point, we may choose the constant $B$ such that $B \geq 3C_g^2$. As long as $C_F > 100 \cdot 2C_g^2$, any $(C_F(k/\epsilon)\log(n/k))$-sized forest $\mathcal{F}$ of prefixes has at least $\frac{99}{100}|\mathcal{F}|$ prefixes which are light, and hence at least \[(99/100)|\mathcal{F}| - (10/100)|\mathcal{F}| = (89/100)|\mathcal{F}|\] prefixes which are both light and desirable (non-undesirable). Lightness and desirability of a prefix $p$ imply that $\sqrt{\tilde{\mathrm{est}}_p} < \frac{1}{C_g^2}\sqrt{\epsilon/k}\|x_{-k}\|_2$. Given the above, the probability that \emph{all} of those prefixes satisfy $\mathrm{est}_p \geq \frac{1}{C_g}\sqrt{\epsilon/k}\|x_{-k}\|_2$ can be calculated as \[\left(e^{-\Omega(R)}\right)^{ (89C_F/100) \cdot (k/\epsilon) \log(n/k)} \leq \frac{\delta}{6}.\]
Thus, the above two conditions and the condition of Claim~\ref{claim:1} hold with probability $1-\frac{\delta}{2}$. Conditioning on the above, the execution of Algorithm~\ref{alg:ell2_decoding} ensures that all but $k/10$ coordinates $i \in \mathcal{H}$ will not be displaced by any other prefix in Line~\ref{lin:ell2_discard}, as long as $C_T > 1 + C_g^2 + C_F$.

This yields the proof of the claim.
\end{proof}

Given the above considerations, we are now ready to prove the lemma. Claim~\ref{claim:2} ensures that with probability $1-\frac{\delta}{2}$ there can be at most $O((k/\epsilon) \log (n/k))$ prefixes $p$ that will have an estimate larger than $\frac{1}{C_g} \sqrt{\epsilon/k}\|x_{-k}\|_2$: i) potentially those not satisfying item~1 of the 
claim, which are $O(k/\epsilon) \cdot H= O((k/\epsilon) \cdot \log (n/k))$ in total, and ii) those who satisfy item~1, which are again $O((k/\epsilon) \log(n/k))$ as proved in the claim. The execution of the algorithm along with Claim~\ref{claim:1} ensures that all but $(k/10)$ elements in $\mathcal{H}$ will be returned in the list $L$ at the end; let those elements constitute the set $T$. A standard argument as in previous work~\cite{glps12,NakosS19} shows that $T$ indeed satisfies the conclusion of the Lemma, finishing the proof.

\end{proof}

\begin{remark}
In fact, we can prove that the list $L$ output by Algorithm~\ref{alg:ell2_decoding} has a size \[O\left( (k/\epsilon) \cdot \frac{\log (n/k)}{ \log \log (n/k)}\right).\] However, this will not make an overall difference in our argument.
\end{remark}

Given the list $L$ from Lemma~\ref{lem:ell2_identify}, we shall prune it down to $2k$ coordinates. The construction is exactly the same as in previous work~\cite{glps12,gnprs13,lnw17}, and the proof of correctness follows from~\cite{lnw17}. The additional sensing matrix we keep is a \textsc{CountSketch} matrix; i.e., a vertical concatenation of \[R' = O(\log(1/\epsilon) + \log (n/k) + \log(1/\delta)/k) = O(\log(n/k) + \log(1/\delta)/k)\] matrices $C^{(1)},\ldots,C^{(R')} \in \{-1,0,+1\}^{O(k) \times n}$, each one having a random sign at one random position in every column. For every $i \in L$, we compute the approximations $x_i' = \mathrm{median}_{r\in[R']} (C^{(r)} x)_{q(r,i)}$, where $q(r,i)$ is the position of the non-zero elements of the $i$-th column in $C^{(r)}$. We then keep the coordinates $i \in L$ with the largest in magnitude approximations. With probability $1-\frac{\delta}{2}$, for all but $(4k/10)$ elements $i \in L$ their approximations satisfy $|x_i - x_i'| \leq c \sqrt{\epsilon/k}\|x_{-k}\|_2$ (good approximation), for some absolute constant $c \ll 1$. Note that in expectation, the amount of \emph{non}-well-approximated coordinates $i \in L$ is at most $e^{-\Omega(R)} \cdot |L|$; using martingale arguments to handle dependency issues (or standard arguments from balls and bins~\cite[Chapter~5]{mitzen}), one can show that the probability that more than $k/10$ coordinates $i\in L$ are not well-approximated is at most $\frac{\delta}{2}$~\cite{lnw17}. As in~\cite{glps12} we may zero out $x_i'$ for every $i$ for which $x_i'$ is not among the top $2k$ coordinates. If we output $x'$ which is $2k$-sparse, we obtain the desired guarantee (for correctness see~\cite{glps12}). This yields an overall failure probability of $\delta$.\newline

The standard way to obtain a scheme for the general $\ell_2/\ell_2$ problem is to vertically pack a sequence of weak systems, associated with matrices $\Psi^{(1)},\Psi^{(2)},\ldots$, with (sparsity, fineness of approximation) \[(k,\epsilon), \left(\frac{k}{2}, \frac{3\epsilon}{2} \right),\left(\frac{k}{4}, \left(\frac{3}{2}\right)^2\cdot\epsilon\right),\ldots,\] respectively. Using the first weak system we may obtain a $(2k)$-sparse vector $x^{(1)}$, and use it to gain access to the sketch $\Psi^{(2)}(x-x^{(1)})$. In turn, this will return vector $x^{(2)}$, which we will feed to the next weak system to gain access to $\Psi^{(3)}(x-x^{(1)} -x^{(2)})$, so on so forth. Choosing $(k_j,\epsilon_j) = (k/2^i, \epsilon/(1.5)^i)$ we may obtain the analogous conversions in~\cite{glps12,hikp12a,gnprs13,glps17,ref:CI17,lnw17}. Other ways of obtaining $\ell_2/\ell_2$ schemes from weak systems appear in~\cite[Section~D.2]{lnw17}.

\section{Conclusion and Future Work}

Sublinear-time sparse recovery is an extensively studied area with several open problems, most of which we seem to resist attacks by the available machinery. Thus, new techniques are needed in order to bypass the current barriers and improve the known trade-offs between the measurement complexity and the decoding time. In this work, we made a major step in this direction by introducing a new sublinear-time algorithmic framework, enabling us to close many of the outstanding open problems and make significant progress on others. It is interesting to consider interactions between our ideas and the list-recovery framework, notably the closely related framework of~\cite{lnnt16}. For list-disjunct matrices, list-recovery based arguments, including~\cite{lnnt16}, are unsuitable on their own for optimal measurement complexity, let alone schemes with near-optimal decoding time. Our intuition is that further progress should be possible by setting up an appropriate error-correcting mechanism on top of our construction. It should also be noted that the constructions of efficiently decodable list-disjunct matrices are significantly more challenging than those for disjunct matrices. This is expected since the former can be used to construct the latter. 

Furthermore, one should first aim at designing very efficient schemes for $\ell_1/\ell_1$ and $\ell_2/\ell_1$ compressed sensing before trying to improve error-correcting group testing schemes (and, in particular, list-disjunct matrices), since the latter seems to be much harder for two reasons: i) the errors are completely adversarial, whereas in compressed sensing they depend on the underlying structure of the vector $x$, and ii) one cannot perform a subtract-and-repeat iterative process and thus has to identify all defective items in one stage, in contrast to compressed sensing where it suffices to identify \emph{most} heavy hitters, and this makes a major difference.




\newpage
\addcontentsline{toc}{section}{References}
\bibliographystyle{customalpha} 
\bibliography{refs}

\newcommand{\etalchar}[1]{$^{#1}$}
\begin{thebibliography}{YAST{\etalchar{+}}20}

\bibitem[ABJ{\etalchar{+}}19]{ABJTW19}
K.~Axiotis, A.~Backurs, C.~Jin, C.~Tzamos, and H.~Wu.
\newblock Fast modular {Subset Sum} using linear sketching.
\newblock In {\em Proceedings of the Thirtieth Annual {ACM-SIAM} Symposium on
  Discrete Algorithms, {SODA}}, pages 58--69. {SIAM}, 2019.

\bibitem[AH09]{alon2009optimal}
N.~Alon and R.~Hod.
\newblock Optimal monotone encodings.
\newblock {\em IEEE Transactions on Information Theory}, 55(3):1343--1353,
  2009.

\bibitem[AIV19]{aamand2019learned}
A.~Aamand, P.~Indyk, and A.~Vakilian.
\newblock (learned) frequency estimation algorithms under {Zipfian}
  distribution.
\newblock {\em arXiv preprint arXiv:1908.05198}, 2019.

\bibitem[AY20]{alman2020faster}
J.~Alman and H.~Yu.
\newblock Faster update time for turnstile streaming algorithms.
\newblock In {\em Proceedings of the Fourteenth Annual ACM-SIAM Symposium on
  Discrete Algorithms}, pages 1803--1813, 2020.

\bibitem[BCI{\etalchar{+}}17]{bcinww17}
V.~Braverman, S.~R. Chestnut, N.~Ivkin, J.~Nelson, Z.~Wang, and D.~P. Woodruff.
\newblock {BPTree}: an $\ell_2$ heavy hitters algorithm using constant memory.
\newblock In {\em Proceedings of the 36th ACM SIGMOD-SIGACT-SIGAI Symposium on
  Principles of Database Systems (PODS)}, pages 361--376, 2017.

\bibitem[BCIW16]{bciw16}
V.~Braverman, S.~R. Chestnut, N.~Ivkin, and D.~P. Woodruff.
\newblock Beating {CountSketch} for heavy hitters in insertion streams.
\newblock In D.~Wichs and Y.~Mansour, editors, {\em Proceedings of the 48th
  Annual {ACM} {SIGACT} Symposium on Theory of Computing, {STOC}}, pages
  740--753, 2016.

\bibitem[BCS{\etalchar{+}}19a]{bondorf2019sublinear}
S.~Bondorf, B.~Chen, J.~Scarlett, H.~Yu, and Y.~Zhao.
\newblock Sublinear-time non-adaptive group testing with {$O (k\log n)$} tests
  via bit-mixing coding.
\newblock {\em arXiv preprint arXiv:1904.10102}, 2019.

\bibitem[BCS{\etalchar{+}}19b]{ref:BCSYZ19}
S.~Bondorf, B.~Chen, J.~Scarlett, H.~Yu, and Y.~Zhao.
\newblock Sublinear-time non-adaptive group testing with {$O(k \log n)$} tests
  via bit-mixing coding.
\newblock {\em CoRR}, abs/1904.10102, 2019.

\bibitem[BDW19]{BDW19}
A.~Bhattacharyya, P.~Dey, and D.~P. Woodruff.
\newblock An optimal algorithm for $\ell_1$-heavy hitters in insertion streams
  and related problems.
\newblock {\em {ACM} Trans. Algorithms}, 15(1):2:1--2:27, 2019.

\bibitem[Ben20]{ref:COVID4}
K.~Bennhold.
\newblock A {German} exception? why the country’s coronavirus death rate is
  low.
\newblock
  \texttt{https://www.nytimes.com/2020/04/04/world/europe/germany-coronavirus-death-
  rate.html}, 2020.
\newblock Accessed: 2020-04-15.

\bibitem[BGI{\etalchar{+}}08]{ref:BGIKS08}
R.~Berinde, A.~Gilbert, P.~Indyk, H.~Karloff, and M.~Strauss.
\newblock Combining geometry and combinatorics: a unified approach to sparse
  signal recovery.
\newblock In {\em Proceedings of the Annual {Allerton} Conference on
  Communication, Control, and Computing}, 2008.
\newblock arXiv:0804.4666.

\bibitem[BGL{\etalchar{+}}18]{BGLWZ18}
V.~Braverman, E.~Grigorescu, H.~Lang, D.~P. Woodruff, and S.~Zhou.
\newblock Nearly optimal distinct elements and heavy hitters on sliding
  windows.
\newblock In {\em Proceedings of {APPROX/RANDOM} 2018}, volume 116 of {\em
  LIPIcs}, pages 7:1--7:22, 2018.

\bibitem[BI11]{bi11}
K.~D. Ba and P.~Indyk.
\newblock Sparse recovery with partial support knowledge.
\newblock In {\em Proceedings of {RANDOM/APPROX}}, pages 26--37, 2011.

\bibitem[BK20]{ref:COVID6}
A.~Z. Broder and R.~Kumar.
\newblock A note on double pooling tests.
\newblock {\em arXiv preprint arXiv:2004.01684}, 2020.

\bibitem[BKB{\etalchar{+}}95]{ref:BKBB95}
W.~Bruno, E.~Knill, D.~Balding, D.~Bruce, N.~Doggett, W.~Sawhill, R.~Stallings,
  C.~Whittaker, and D.~Torney.
\newblock Efficient pooling designs for library screening.
\newblock {\em Genomics}, 26(1):21--30, 1995.

\bibitem[BN20]{BN20}
K.~Bringmann and V.~Nakos.
\newblock Top-k-convolution and the quest for near-linear output-sensitive
  subset sum.
\newblock In K.~Makarychev, Y.~Makarychev, M.~Tulsiani, G.~Kamath, and
  J.~Chuzhoy, editors, {\em Proccedings of the 52nd Annual {ACM} {SIGACT}
  Symposium on Theory of Computing, {STOC} 2020, Chicago, IL, USA, June 22-26,
  2020}, pages 982--995. {ACM}, 2020.

\bibitem[CCF02]{ccf02}
M.~Charikar, K.~Chen, and M.~{Farach-Colton}.
\newblock Finding frequent items in data streams.
\newblock In {\em Automata, Languages and Programming}, pages 693--703.
  Springer, 2002.

\bibitem[CD06]{ref:handbook}
C.~J. Colbourn and J.~H. Dinitz.
\newblock {\em Handbook of Combinatorial Designs, Second Edition (Discrete
  Mathematics and Its Applications)}.
\newblock Chapman \& Hall/CRC, 2006.

\bibitem[CD08]{cheng2008new}
Y.~Cheng and D.-Z. Du.
\newblock New constructions of one-and two-stage pooling designs.
\newblock {\em Journal of Computational Biology}, 15(2):195--205, 2008.

\bibitem[CDH07]{ref:CDH07}
H.-B. Chen, D.-Z. Du, and F.-K. Hwang.
\newblock An unexpected meeting of four seemingly unrelated problems: graph
  testing, {DNA} complex screening, superimposed codes and secure key
  distribution.
\newblock {\em Journal of Combinatorial Optimization}, 14(2-3):121--129, 2007.

\bibitem[CEPR07]{ref:CEPR07}
R.~Clifford, K.~Efremenko, E.~Porat, and A.~Rothschild.
\newblock $k$-mismatch with don't cares.
\newblock In {\em Proceedings of the $15$th European Symposium on Algorithm
  ({ESA})}, volume 4698 of {\em Lecture Notes in Computer Science}, pages
  151--162, 2007.

\bibitem[CG{\"O}R00]{ref:CGOR00}
B.~S. Chlebus, L.~G{\c{c}}asieniec, A.~{\"O}stlin, and J.~M. Robson.
\newblock Deterministic radio broadcasting.
\newblock In U.~Montanari, J.~D.~P. Rolim, and E.~Welzl, editors, {\em
  Automata, Languages and Programming}, pages 717--729. Springer Berlin
  Heidelberg, 2000.

\bibitem[CH08]{chen2008survey}
H.-B. Chen and F.~K. Hwang.
\newblock A survey on nonadaptive group testing algorithms through the angle of
  decoding.
\newblock {\em Journal of Combinatorial Optimization}, 15(1):49--59, 2008.

\bibitem[CH09]{ch09}
G.~Cormode and M.~Hadjieleftheriou.
\newblock Finding the frequent items in streams of data.
\newblock {\em Communications of the ACM}, 52(10):97--105, 2009.

\bibitem[Che12]{ch2009noise}
M.~Cheraghchi.
\newblock Noise-resilient group testing: Limitations and constructions.
\newblock {\em Discrete Applied Mathematics}, 161(1--2):81--95, 2012.
\newblock Preliminary version in Proceedings of {FCT}, LNCS:5699, pp.~62--73,
  2009, arXiv version (arXiv:0811.2609) in 2008.

\bibitem[Che13]{Cheraghchi13}
M.~Cheraghchi.
\newblock Improved constructions for non-adaptive threshold group testing.
\newblock {\em Algorithmica}, 67(3):384--417, 2013.

\bibitem[CI17]{ref:CI17}
M.~Cheraghchi and P.~Indyk.
\newblock Nearly optimal deterministic algorithm for sparse {Walsh-Hadamard
  Transform}.
\newblock {\em ACM Transactions on Algorithms}, 13(3), 2017.

\bibitem[CJBJ13]{cai2013grotesque}
S.~Cai, M.~Jahangoshahi, M.~Bakshi, and S.~Jaggi.
\newblock {GROTESQUE}: noisy group testing (quick and efficient).
\newblock In {\em Proceedings of the 51st Annual Allerton Conference on
  Communication, Control, and Computing}, pages 1234--1241. IEEE, 2013.

\bibitem[CJBJ17]{cai2017efficient}
S.~Cai, M.~Jahangoshahi, M.~Bakshi, and S.~Jaggi.
\newblock Efficient algorithms for noisy group testing.
\newblock {\em IEEE Transactions on Information Theory}, 63(4):2113--2136,
  2017.

\bibitem[CKMS12]{ref:GCGT12}
M.~{Cheraghchi}, A.~{Karbasi}, S.~{Mohajer}, and V.~{Saligrama}.
\newblock Graph-constrained group testing.
\newblock {\em IEEE Transactions on Information Theory}, 58(1):248--262, 2012.

\bibitem[CKSZ17]{cksz17}
V.~Cevher, M.~Kapralov, J.~Scarlett, and A.~Zandieh.
\newblock An adaptive sublinear-time block sparse {F}ourier transform.
\newblock In {\em Proceedings of the 49th Annual Symposium on the Theory of
  Computing (STOC)}. ACM, 2017.

\bibitem[CM04]{cm04}
G.~Cormode and S.~Muthukrishnan.
\newblock An improved data stream summary: the count-min sketch and its
  applications.
\newblock In {\em Proceedings of the Annual Conference on Foundations of
  Software Technology and Theoretical Computer Science}, 2004.

\bibitem[CM05]{ref:CM05}
G.~Cormode and S.~Muthukrishnan.
\newblock What's hot and what's not: tracking most frequent items dynamically.
\newblock {\em ACM Transactions on Database Systems}, 30(1):249--278, 2005.

\bibitem[CR19]{ref:CR19b}
M.~Cheraghchi and J.~Ribeiro.
\newblock Simple codes and sparse recovery with fast decoding.
\newblock In {\em Proceedings of the {IEEE International Symposium on
  Information Theory (ISIT)}}, pages 156--160, 2019.

\bibitem[CT06]{CandesTao}
E.~J. Candes and T.~Tao.
\newblock Near-optimal signal recovery from random projections: Universal
  encoding strategies?
\newblock {\em {IEEE} Trans. Inf. Theory}, 52(12):5406--5425, 2006.

\bibitem[DH00a]{ref:groupTesting}
D.-Z. Du and F.-K. Hwang.
\newblock {\em Combinatorial Group Testing and its Applications}.
\newblock World Scientific, second edition, 2000.

\bibitem[DH00b]{du2000combinatorial}
D.-Z. Du and F.-K. Hwang.
\newblock {\em Combinatorial group testing and its applications}, volume~12.
\newblock World Scientific, 2000.

\bibitem[DH06]{ref:DH06}
D.-Z. Du and F.-K. Hwang.
\newblock {\em Pooling Designs and Nonadaptive Group Testing}.
\newblock World Scientific, 2006.

\bibitem[Don06]{d06}
D.~L. Donoho.
\newblock Compressed sensing.
\newblock {\em {IEEE} Trans. Information Theory}, 52(4):1289--1306, 2006.

\bibitem[Dor43]{ref:Dor43}
R.~Dorfman.
\newblock The detection of defective members of large populations.
\newblock {\em Annals of Mathematical Statistics}, 14:436--440, 1943.

\bibitem[DR82]{d1982bounds}
A.~G. D'yachkov and V.~V. Rykov.
\newblock Bounds on the length of disjunctive codes.
\newblock {\em Problemy Peredachi Informatsii}, 18(3):7--13, 1982.

\bibitem[DR83]{dyachkov1983survey}
A.~G. D'yachkov and V.~V. Rykov.
\newblock A survey of superimposed code theory.
\newblock {\em Problems of Control and Information Theory}, 12(4):1--13, 1983.

\bibitem[DRR89]{dyachkov1989superimposed}
A.~D'yachkov, V.~Rykov, and A.~Rashad.
\newblock Superimposed distance codes.
\newblock {\em Problems of Control and Information Theory-problemy Upravleniya
  i Teorii Informatsii}, 18(4):237--250, 1989.

\bibitem[Eur20]{ref:COVID2}
EurekAlert.
\newblock Pool testing of {SARS-CoV-02} samples increases worldwide test
  capacities many times over.
\newblock \texttt{https://www.eurekalert.org/pub\_releases/2020-03/
  guf-pto033020.php}, 2020.
\newblock Accessed: 2020-04-15.

\bibitem[FKKM97]{ref:KKM97}
M.~Farach, S.~Kannan, E.~Knill, and S.~Muthukrishnan.
\newblock Group testing problems with sequences in experimental molecular
  biology.
\newblock In {\em Proceedings of Compression and Complexity of Sequences},
  pages 357--367, 1997.

\bibitem[FR13]{ref:CSbook}
S.~Foucart and H.~Rauhut.
\newblock {\em A Mathematical Introduction to Compressive Sensing}.
\newblock Birkh\"{a}user Basel, 2013.

\bibitem[GAT05]{goodrich2005indexing}
M.~T. Goodrich, M.~J. Atallah, and R.~Tamassia.
\newblock Indexing information for data forensics.
\newblock In {\em International Conference on Applied Cryptography and Network
  Security}, pages 206--221. Springer, 2005.

\bibitem[GGJZ19]{GGJZ19}
V.~Gandikota, E.~Grigorescu, S.~Jaggi, and S.~Zhou.
\newblock Nearly optimal sparse group testing.
\newblock {\em {IEEE} Trans. Inf. Theory}, 65(5):2760--2773, 2019.

\bibitem[GI04]{GuruswamiI04}
V.~Guruswami and P.~Indyk.
\newblock Linear-time list decoding in error-free settings: (extended
  abstract).
\newblock In {\em Proceedings of {ICALP}}, volume 3142 of {\em Lecture Notes in
  Computer Science}, pages 695--707. Springer, 2004.

\bibitem[GI10]{gi10}
A.~Gilbert and P.~Indyk.
\newblock Sparse recovery using sparse matrices.
\newblock {\em Proceedings of the IEEE}, 98(6):937--947, 2010.

\bibitem[GLPS10]{glps12}
A.~C. Gilbert, Y.~Li, E.~Porat, and M.~J. Strauss.
\newblock Approximate sparse recovery: optimizing time and measurements.
\newblock {\em SIAM Journal on Computing 2012}, 41(2):436--453, 2010.

\bibitem[GLPS17]{glps17}
A.~C. Gilbert, Y.~Li, E.~Porat, and M.~J. Strauss.
\newblock For-all sparse recovery in near-optimal time.
\newblock {\em ACM Transactions on Algorithms (TALG)}, 13(3):32, 2017.

\bibitem[GMS05]{gms05}
A.~C. Gilbert, S.~Muthukrishnan, and M.~Strauss.
\newblock Improved time bounds for near-optimal sparse {F}ourier
  representations.
\newblock In {\em Optics \& Photonics 2005}, pages 398--412. International
  Society for Optics and Photonics, 2005.

\bibitem[GNP{\etalchar{+}}13]{gnprs13}
A.~C. Gilbert, H.~Q. Ngo, E.~Porat, A.~Rudra, and M.~J. Strauss.
\newblock $\ell_2/\ell_2$-foreach sparse recovery with low risk.
\newblock In {\em International Colloquium on Automata, Languages, and
  Programming}, pages 461--472. Springer, 2013.

\bibitem[GRS19]{ref:GRS19}
V.~Guruswami, A.~Rudra, and M.~Sudan.
\newblock Essential coding theory.
\newblock \texttt{https://
  cse.buffalo.edu/faculty/atri/courses/coding-theory/book/}, 2019.
\newblock Accessed: 2020-04-15.

\bibitem[GSTV06]{ref:GSTV06}
A.~Gilbert, M.~Strauss, J.~Tropp, and R.~Vershynin.
\newblock Algorithmic linear dimension reduction in the $\ell_1$ norm for
  sparse vectors.
\newblock In {\em Proceedings of the Annual {Allerton} Conference on
  Communication, Control, and Computing}, 2006.

\bibitem[HIKP12a]{hikp12a}
H.~Hassanieh, P.~Indyk, D.~Katabi, and E.~Price.
\newblock Nearly optimal sparse {F}ourier transform.
\newblock In {\em Proceedings of the forty-fourth annual ACM symposium on
  Theory of Computing}, pages 563--578. ACM, 2012.

\bibitem[HIKP12b]{hikp12b}
H.~Hassanieh, P.~Indyk, D.~Katabi, and E.~Price.
\newblock Simple and practical algorithm for sparse {F}ourier transform.
\newblock In {\em Proceedings of the twenty-third annual ACM-SIAM symposium on
  Discrete Algorithms}, pages 1183--1194. SIAM, 2012.

\bibitem[HL00]{ref:HL00}
E.-S. Hong and R.~Ladner.
\newblock Group testing for image compression.
\newblock In {\em Data Compression Conference}, pages 3--12, 2000.

\bibitem[HNO08]{HNO08}
N.~J.~A. Harvey, J.~Nelson, and K.~Onak.
\newblock Sketching and streaming entropy via approximation theory.
\newblock In {\em Proceedings of the 49th Annual {IEEE} Symposium on
  Foundations of Computer Science, ({FOCS})}, pages 489--498, 2008.

\bibitem[IK14]{ik14}
P.~Indyk and M.~Kapralov.
\newblock Sample-optimal {F}ourier sampling in any constant dimension.
\newblock In {\em Proceedings of the Annual IEEE Symposium on Foundations of
  Computer Science ({FOCS})}, pages 514--523, 2014.

\bibitem[IKW{\"O}19]{inan2019optimality}
H.~A. Inan, P.~Kairouz, M.~Wootters, and A.~{\"O}zg{\"u}r.
\newblock On the optimality of the {Kautz-Singleton} construction in
  probabilistic group testing.
\newblock {\em IEEE Transactions on Information Theory}, 2019.

\bibitem[INR10]{indyk2010efficiently}
P.~Indyk, H.~Q. Ngo, and A.~Rudra.
\newblock Efficiently decodable non-adaptive group testing.
\newblock In {\em Proceedings of the twenty-first annual ACM-SIAM symposium on
  Discrete Algorithms}, pages 1126--1142. SIAM, 2010.

\bibitem[Iwe13]{iw13}
M.~A. Iwen.
\newblock Improved approximation guarantees for sublinear-time {F}ourier
  algorithms.
\newblock {\em Applied And Computational Harmonic Analysis}, 34(1):57--82,
  2013.

\bibitem[JW18]{JW18}
R.~Jayaram and D.~P. Woodruff.
\newblock Perfect lp sampling in a data stream.
\newblock In M.~Thorup, editor, {\em 59th {IEEE} Annual Symposium on
  Foundations of Computer Science, {FOCS} 2018, Paris, France, October 7-9,
  2018}, pages 544--555, 2018.

\bibitem[Kap16]{k16}
M.~Kapralov.
\newblock Sparse {F}ourier transform in any constant dimension with
  nearly-optimal sample complexity in sublinear time.
\newblock In {\em Symposium on Theory of Computing Conference, STOC'16,
  Cambridge, MA, USA, June 19-21, 2016}, 2016.

\bibitem[Kap17]{k17}
M.~Kapralov.
\newblock Sample efficient estimation and recovery in sparse {FFT} via
  isolation on average.
\newblock In {\em Proceedings of the Annual IEEE Symposium on Foundations of
  Computer Science ({FOCS})}, 2017.

\bibitem[KM95]{ref:KM95}
E.~Knill and S.~Muthukrishnan.
\newblock Group testing problems in experimental molecular biology.
\newblock {\em arXiv preprint arXiv:math/9505211}, 1995.

\bibitem[KNPW11]{KaneNPW11}
D.~M. Kane, J.~Nelson, E.~Porat, and D.~P. Woodruff.
\newblock Fast moment estimation in data streams in optimal space.
\newblock In {\em Proceedings of the 43rd {ACM} Symposium on Theory of
  Computing, {STOC} 2011, San Jose, CA, USA, 6-8 June 2011}, pages 745--754.
  {ACM}, 2011.

\bibitem[KNW10]{KNW10}
D.~M. Kane, J.~Nelson, and D.~P. Woodruff.
\newblock On the exact space complexity of sketching and streaming small norms.
\newblock In M.~Charikar, editor, {\em Proceedings of the Twenty-First Annual
  {ACM-SIAM} Symposium on Discrete Algorithms, {SODA} 2010, Austin, Texas, USA,
  January 17-19, 2010}, pages 1161--1178. {SIAM}, 2010.

\bibitem[KS64]{ref:KS64}
W.~Kautz and R.~Singleton.
\newblock Nonrandom binary superimposed codes.
\newblock {\em IEEE Transactions on Information Theory}, 10:363--377, 1964.

\bibitem[KVZ19]{kvz19}
M.~Kapralov, A.~Velingker, and A.~Zandieh.
\newblock Dimension-independent sparse {F}ourier transform.
\newblock In {\em Proceedings of the Thirtieth Annual ACM-SIAM Symposium on
  Discrete Algorithms (SODA)}, pages 2709--2728, 2019.

\bibitem[LCPR19]{lee2019saffron}
K.~Lee, K.~Chandrasekher, R.~Pedarsani, and K.~Ramchandran.
\newblock Saffron: A fast, efficient, and robust framework for group testing
  based on sparse-graph codes.
\newblock {\em IEEE Transactions on Signal Processing}, 2019.

\bibitem[LN18a]{nakos2017deterministic_heavy_hitters}
Y.~Li and V.~Nakos.
\newblock Deterministic heavy hitters with sublinear query time.
\newblock In {\em Approximation, Randomization, and Combinatorial Optimization.
  Algorithms and Techniques, {APPROX/RANDOM} 2018, August 20-22, 2018 -
  Princeton, NJ, {USA}}, pages 18:1--18:18, 2018.

\bibitem[LN18b]{LiN18}
Y.~Li and V.~Nakos.
\newblock Sublinear-time algorithms for compressive phase retrieval.
\newblock In {\em 2018 {IEEE} International Symposium on Information Theory,
  {ISIT} 2018, Vail, CO, USA, June 17-22, 2018}, pages 2301--2305. {IEEE},
  2018.

\bibitem[LNNT16]{lnnt16}
K.~G. Larsen, J.~Nelson, H.~L. Nguy\stackon[2pt]{\^{e}}{\~{}}n, and M.~Thorup.
\newblock Heavy hitters via cluster-preserving clustering.
\newblock In {\em Foundations of Computer Science (FOCS), 2016 IEEE 57th Annual
  Symposium on}, pages 61--70. IEEE, 2016.

\bibitem[LNW18]{lnw17}
Y.~Li, V.~Nakos, and D.~P. Woodruff.
\newblock On low-risk heavy hitters and sparse recovery schemes.
\newblock In {\em Approximation, Randomization, and Combinatorial Optimization.
  Algorithms and Techniques, {APPROX/RANDOM} 2018, August 20-22, 2018 -
  Princeton, NJ, {USA}}, pages 19:1--19:13, 2018.

\bibitem[Mac99]{ref:Mac99}
A.~Macula.
\newblock Probabilistic nonadaptive group testing in the presence of errors and
  {DNA} library screening.
\newblock {\em Annals of Combinatorics}, 3(1):61--69, 1999.

\bibitem[MNS07]{moran2007deterministic}
T.~Moran, M.~Naor, and G.~Segev.
\newblock Deterministic history-independent strategies for storing information
  on write-once memories.
\newblock In {\em International Colloquium on Automata, Languages, and
  Programming}, pages 303--315. Springer, 2007.

\bibitem[MU17]{mitzen}
M.~Mitzenmacher and E.~Upfal.
\newblock {\em Probability and computing: Randomization and probabilistic
  techniques in algorithms and data analysis}.
\newblock Cambridge university press, 2017.

\bibitem[Nak17a]{Nakos17}
V.~Nakos.
\newblock Almost optimal phaseless compressed sensing with sublinear decoding
  time.
\newblock In {\em 2017 {IEEE} International Symposium on Information Theory,
  {ISIT} 2017, Aachen, Germany, June 25-30, 2017}, pages 1142--1146. {IEEE},
  2017.

\bibitem[Nak17b]{NakosOne17}
V.~Nakos.
\newblock On fast decoding of high-dimensional signals from one-bit
  measurements.
\newblock In {\em 44th International Colloquium on Automata, Languages, and
  Programming, {ICALP} 2017, July 10-14, 2017, Warsaw, Poland}, volume~80 of
  {\em LIPIcs}, pages 61:1--61:14, 2017.

\bibitem[ND00]{ref:ND00}
H.-Q. Ngo and D.-Z. Du.
\newblock A survey on combinatorial group testing algorithms with applications
  to {DNA} library screening.
\newblock {\em DIMACS Series on Discrete Math.\ and Theoretical Computer
  Science}, 55:171--182, 2000.

\bibitem[NHL20]{nhl20}
K.~R. Narayanan, A.~Heidarzadeh, and R.~Laxminarayan.
\newblock On accelerated testing for {COVID-19} using group testing.
\newblock {\em arXiv preprint arXiv:2004.04785}, 2020.

\bibitem[NN13]{nelson2013sparsity}
J.~Nelson and H.~L. Nguy\stackon[2pt]{\^{e}}{\~{}}n.
\newblock Sparsity lower bounds for dimensionality reducing maps.
\newblock In {\em Proceedings of the forty-fifth annual ACM symposium on Theory
  of Computing}, pages 101--110, 2013.

\bibitem[NNW12]{NNW12}
J.~Nelson, H.~L. Nguy\stackon[2pt]{\^{e}}{\~{}}n, and D.~P. Woodruff.
\newblock On deterministic sketching and streaming for sparse recovery and norm
  estimation.
\newblock In {\em Approximation, Randomization, and Combinatorial Optimization.
  Algorithms and Techniques - 15th International Workshop, {APPROX} 2012, and
  16th International Workshop, {RANDOM} 2012, Cambridge, MA, USA, August 15-17,
  2012. Proceedings}, volume 7408 of {\em Lecture Notes in Computer Science},
  pages 627--638. Springer, 2012.

\bibitem[NPR11]{icalp11}
H.~Q. Ngo, E.~Porat, and A.~Rudra.
\newblock Efficiently decodable error-correcting list disjunct matrices and
  applications.
\newblock In {\em International Colloquium on Automata, Languages, and
  Programming}, pages 557--568. Springer, 2011.

\bibitem[NPRR18]{nprr18}
H.~Q. Ngo, E.~Porat, C.~R{\'{e}}, and A.~Rudra.
\newblock Worst-case optimal join algorithms.
\newblock {\em J. {ACM}}, 65(3):16:1--16:40, 2018.

\bibitem[NS19]{NakosS19}
V.~Nakos and Z.~Song.
\newblock Stronger $\ell_2/\ell_2$ compressed sensing; without iterating.
\newblock In {\em Proceedings of the 51st Annual {ACM} {SIGACT} Symposium on
  Theory of Computing, {STOC} 2019, Phoenix, AZ, USA, June 23-26, 2019}, pages
  289--297. {ACM}, 2019.

\bibitem[NSW19]{NakosSW19}
V.~Nakos, Z.~Song, and Z.~Wang.
\newblock (nearly) sample-optimal sparse {Fourier Transform} in any dimension;
  {RIPless} and filterless.
\newblock In {\em 60th {IEEE} Annual Symposium on Foundations of Computer
  Science, {FOCS} 2019, Baltimore, Maryland, USA, November 9-12, 2019}, pages
  1568--1577, 2019.

\bibitem[{Oma}20]{ref:COVID3}
{Omaha News}.
\newblock Gov.\ {Ricketts} provides update on coronavirus testing.
\newblock \texttt{https://
  www.3newsnow.com/news/coronavirus/live-gov-ricketts-provides-
  coronavirus-briefing-3-24-20}, 2020.
\newblock Accessed: 2020-04-15.

\bibitem[PR08]{porat2008explicit}
E.~Porat and A.~Rothschild.
\newblock Explicit non-adaptive combinatorial group testing schemes.
\newblock In {\em International Colloquium on Automata, Languages, and
  Programming}, pages 748--759. Springer, 2008.

\bibitem[PS12]{ps12}
E.~Porat and M.~J. Strauss.
\newblock Sublinear time, measurement-optimal, sparse recovery for all.
\newblock In {\em Proceedings of the twenty-third annual ACM-SIAM symposium on
  Discrete Algorithms}, pages 1215--1227. Society for Industrial and Applied
  Mathematics, 2012.

\bibitem[PS20]{PriceS20}
E.~Price and J.~Scarlett.
\newblock A fast binary splitting approach to non-adaptive group testing.
\newblock In J.~Byrka and R.~Meka, editors, {\em Approximation, Randomization,
  and Combinatorial Optimization. Algorithms and Techniques, {APPROX/RANDOM}
  2020, August 17-19, 2020, Virtual Conference}, volume 176 of {\em LIPIcs},
  pages 13:1--13:20. Schloss Dagstuhl - Leibniz-Zentrum f{\"{u}}r Informatik,
  2020.

\bibitem[Rus94]{ruszinko1994upper}
M.~Ruszink{\'o}.
\newblock On the upper bound of the size of the $r$-cover-free families.
\newblock {\em Journal of Combinatorial Theory, Series A}, 66(2):302--310,
  1994.

\bibitem[Sch96]{treecodes91}
L.~J. Schulman.
\newblock Coding for interactive communication.
\newblock {\em IEEE transactions on information theory}, 42(6):1745--1756,
  1996.

\bibitem[SG59]{ref:SG59}
M.~Sobel and P.~Groll.
\newblock Group-testing to eliminate efficiently all defectives in a binomial
  sample.
\newblock {\em Bell Systems Technical Journal}, 38:1179--1252, 1959.

\bibitem[Slo07]{sloane2007line}
N.~J. Sloane.
\newblock The on-line encyclopedia of integer sequences.
\newblock In {\em International Conference on Mathematical Knowledge
  Management}, pages 130--130. Springer, 2007.

\bibitem[Spi96]{s96}
D.~A. Spielman.
\newblock Linear-time encodable and decodable error-correcting codes.
\newblock In {\em IEEE Transactions on Information Theory (A preliminary
  version of this paper appears in STOC 1995)}, volume 42:6, pages 1723--1731.
  IEEE, 1996.

\bibitem[STR03]{ref:STR03}
A.~Schliep, D.~Torney, and S.~Rahmann.
\newblock Group testing with {DNA} chips: Generating designs and decoding
  experiments.
\newblock In {\em Proceedings of Computational Systems Bioinformatics}, 2003.

\bibitem[SW19]{SpangW19}
B.~Spang and M.~Wootters.
\newblock Unconstraining graph-constrained group testing.
\newblock In {\em Approximation, Randomization, and Combinatorial Optimization.
  Algorithms and Techniques, {APPROX/RANDOM}}, volume 145 of {\em LIPIcs},
  pages 46:1--46:20, 2019.

\bibitem[Tec20]{ref:COVID1}
Technion.
\newblock Pooling method for accelerated testing of {COVID-19}.
\newblock \texttt{https://www.
  technion.ac.il/en/2020/03/pooling-method-for-accelerated-testing-of-covid
  -19}, 2020.
\newblock Accessed: 2020-04-15.

\bibitem[Vad10]{ref:Vad10}
S.~Vadhan.
\newblock The unified theory of pseudorandomness.
\newblock In {\em Proceedings of the International Congress of Mathematicians},
  2010.

\bibitem[vdBLSS20]{lp20}
J.~van~den Brand, Y.~T. Lee, A.~Sidford, and Z.~Song.
\newblock Solving tall dense linear programs in nearly linear time.
\newblock {\em STOC 2020, To appear}, 2020.

\bibitem[VJN17]{vem2017group}
A.~Vem, N.~T. Janakiraman, and K.~R. Narayanan.
\newblock Group testing using left-and-right-regular sparse-graph codes.
\newblock {\em arXiv preprint arXiv:1701.07477}, 2017.

\bibitem[vzGG13]{algebra}
J.~von~zur Gathen and J.~Gerhard.
\newblock {\em Modern computer algebra}.
\newblock Cambridge university press, 2013.

\bibitem[WHHL06]{ref:WHL06}
W.~Wu, Y.~Huang, X.~Huang, and Y.~Li.
\newblock On error-tolerant {DNA} screening.
\newblock {\em Discrete Applied Mathematics}, 154(12):1753--1758, 2006.

\bibitem[WLHD08]{ref:WLHD08}
W.~Wu, Y.~Li, C.~Huang, and D.~Du.
\newblock Molecular biology and pooling design.
\newblock {\em Data Mining in Biomedicine}, 7:133--139, 2008.

\bibitem[Wol85]{ref:Wol85}
J.~Wolf.
\newblock Born-again group testing: multiaccess communications.
\newblock {\em IEEE Transactions on Information Theory}, 31:185--191, 1985.

\bibitem[Woo16]{w16}
D.~P. Woodruff.
\newblock New algorithms for heavy hitters in data streams.
\newblock {\em arXiv preprint arXiv:1603.01733}, 2016.

\bibitem[YAST{\etalchar{+}}20]{ref:COVID7}
I.~Yelin, N.~Aharony, E.~Shaer-Tamar, A.~Argoetti, E.~Messer, D.~Berenbaum,
  E.~Shafran, A.~Kuzli, N.~Gandali, T.~Hashimshony, Y.~Mandel-Gutfreund,
  M.~Halberthal, Y.~Geffen, M.~Szwarcwort-Cohen, and R.~Kishony.
\newblock Evaluation of {COVID-19 RT-qPCR} test in multi-sample pools.
\newblock {\em medRxiv}, 2020.

\bibitem[ZRB20]{ref:COVID5}
J.~Zhu, K.~Rivera, and D.~Baron.
\newblock Noisy pooled {PCR} for virus testing.
\newblock {\em arXiv preprint arXiv:2004.02689}, 2020.

\end{thebibliography}
\newpage




\end{document}